\documentclass[a4paper, 12pt]{article}
\usepackage{authblk}
\usepackage{amsmath}
\usepackage{mathtools}
\usepackage{setspace}
\usepackage{algorithm}
\usepackage{algpseudocode}
\usepackage{amsfonts}
\usepackage{enumitem}
\usepackage{xcolor}
\definecolor{ForestGreen}{RGB}{34,139,34}
\usepackage{bm}
\usepackage{comment}
\usepackage{subcaption}
\usepackage[comma,authoryear,round]{natbib}
\usepackage{booktabs}
\usepackage{threeparttable}
\usepackage{graphicx}
\usepackage{bibunits}
\usepackage[a4paper,bottom=0.8in,top=0.8in]{geometry}
\usepackage{tikz}
\usetikzlibrary{shapes,decorations,arrows,calc,arrows.meta,fit,positioning}
\tikzset{
    -Latex,auto,node distance =1 cm and 1 cm,semithick,
    state/.style ={ellipse, draw, minimum width = 0.7 cm},
    point/.style = {circle, draw, inner sep=0.04cm,fill,node contents={}},
    bidirected/.style={Latex-Latex,dashed},
    el/.style = {inner sep=2pt, align=left, sloped}
}

\usepackage{amsthm}

\theoremstyle{definition}
\newtheorem{definition}{Definition}

\theoremstyle{plain}
\newtheorem{theorem}{Theorem}

\newtheorem{assumption}{Assumption}
\newtheorem{Aassumption}{Assumption}
\newtheorem{bassumption}{Assumption}

\newtheorem{condition}{Condition}

\newtheorem{lemma}{Lemma}
\newtheorem{remark}{Remark}
\newtheorem{Aremark}{Remark}

\newtheorem{corollary}{Corollary}

\newcommand{\indep}{\perp \!\!\! \perp}
\usepackage{indentfirst}
\usepackage{caption}
\RequirePackage[colorlinks,citecolor=blue,urlcolor=blue, linkcolor = blue]{hyperref}
\usepackage{nameref}
\begin{document} 
\title{\textbf{\Large Causal Inference for Unobservable Multivariate Outcomes, with Applications to Brain Effective Connectivity}}
\author[1]{\large Haiyue Song}
\author[1]{\large Ani Eloyan}
\author[,1]{\large Youjin Lee\footnote{Correspondence: Youjin Lee, Department of Biostatistics, Brown University, Providence, RI 02903, USA. youjin\_lee@brown.edu}}
\author[ ]{\\ \large for the Alzheimer’s Disease Neuroimaging Initiative\footnote{Data used in preparation of this article were obtained from the Alzheimer's Disease Neuroimaging Initiative (ADNI) database (adni.loni.usc.edu). As such, the investigators within the ADNI contributed to the design and implementation of ADNI and/or provided data but did not participate in the analysis or writing of this report. A complete listing of ADNI investigators can be found at: http://adni.loni.usc.edu/wp-content/uploads/how\_to\_apply/ADNI\_Acknowledgement\_List.pdf}}
\affil[1]{\normalsize Department of Biostatistics, Brown University, Providence, RI 02903, USA}
\date{}
\maketitle
\vspace{-2.2em}
\doublespacing
\begin{abstract}
Evaluating the causal effect of an intervention on multivariate outcomes is challenging when the outcomes are interdependent and derived rather than directly observed. Effective connectivity, which summarizes the directional neural communication between brain regions, is one such derived relational outcome. Estimating how external interventions affect effective connectivity introduces two layers of causal inference problems: identifying directional relationships among brain regions from high-dimensional neuroimaging time series and estimating the causal effect of the intervention on these derived relationships. Each layer introduces distinct biases. The first arises from within-outcome dependencies unrelated to the intervention; to address this, we propose a sample-splitting method for estimating meaningful, and potentially causally informative, effective connectivity measures. The second arises from confounding between the intervention and the derived outcomes; to address this, we apply inverse probability weighting methods and incorporate multiple testing when causal effects on multiple components of the outcomes are of interest. We demonstrate, through theoretical results and simulations, that the proposed methods are asymptotically valid under certain conditions with effective type-I and familywise error control. Finally, we apply the proposed methods to examine the causal effect of amyloid on effective connectivity using the resting-state fMRI data from the Alzheimer’s Disease Neuroimaging Initiative database.
\end{abstract}

\section{Introduction} \label{section1: introduction}

Evaluating the causal effect of an intervention on multivariate outcomes is challenging, particularly when inference targets specific components of the outcome rather than an overall outcome. In such settings, valid inference requires accounting not only for potential confounding between the intervention and the outcomes, but also for dependence among outcome components that may induce spurious associations. For example, consider the causal effect of an external intervention on functional synchrony of the human brain, as quantified by \textit{effective connectivity} between pairs of brain regions. Effective connectivity refers to the directional interactions among brain regions~\citep{friston_functional_1994,  stephan_analyzing_2010}. In this context, the inferential goal is to examine how an intervention alters these directional interaction patterns, while accounting for the fact that baseline connectivity may exist even in the absence of intervention. A key challenge, therefore, is to disentangle dependence among outcome components  to obtain effective connectivity measures valid for downstream causal inference, while also addressing confounding between different intervention groups.

A further challenge in causal inference with such multivariate, relational outcomes is that the outcomes of interest are often defined in an abstract way or in a measure-specific manner. In many applications, outcomes of interest are not directly observed but are constructed through model-based approaches. For instance, effective connectivity is typically inferred using parametric or semi-parametric models, yet there is no generally accepted framework for defining or estimating directional interactions among brain regions. Similar issues arise in other domains, including psychology~\citep{vanderweele2022constructed, gilbert2025measurement}, political science~\citep{knox2022testing, stoetzer2025causal}, and genomics~\citep{lozano2009grouped}. In such settings, causal inference proceeds on outcomes that are themselves estimated, potentially introducing additional layers of uncertainty and bias into causal effect estimation.

The overarching goal of this study is to develop a novel method for examining causal effects of an intervention on multivariate, relational outcomes derived from time-series data. The setting presents two important inferential challenges. First, dependence among outcome components, some of which may reflect underlying causal relationships, can induce spurious associations between the intervention and individual outcome components. Second, the relationships between the intervention and the multivariate outcomes may be subject to confounding. We propose a framework that addresses both challenges for valid causal inference for interdependent, derived multivariate outcomes.  

\subsection{Motivating application}

As a motivating application, we consider the causal effect of brain amyloid positivity on effective connectivity. Amyloid plaque accumulation is an important pathological feature of Alzheimer’s disease (AD)~\citep{alzheimer2025}, and {\it in vivo} amyloid burden measured using positron emission tomography (PET) is a widely used biomarker for disease diagnosis and monitoring of disease progression. Prior studies have reported associations between increased amyloid burden and alterations in brain networks and connectivity patterns~\citep{palmqvist2017earliest, moffat2022unravelling}. Notably, amyloid positivity in this study is determined based on PET imaging and is distinct from the our outcomes from functional neuroimaging data.  

We use data from the Alzheimer’s Disease Neuroimaging Initiative (ADNI) available at {\texttt adni.loni.usc.edu}, a multi-center longitudinal study established to facilitate research on AD and the development and validation of disease biomarkers. Our analysis focuses on participants from ANDI phases 3 and 4, during which both amyloid PET and resting-state functional magnetic resonance imaging (fMRI) data were collected. Resting-state fMRI captures temporal variation in blood-oxygen-level-dependent (BOLD) signals across the brain. Following standard ADNI preprocessing protocols, we extract regional time series by aggregating voxel-level signals within prespecified regions of interest (ROI). In addition to imaging data, baseline covariates commonly associated with amyloid burden and brain connectivity are available, including sex, age, years of education, and apolipoprotein E (APOE) $\varepsilon4$ carrier status.

This application raises several fundamental questions. First, how can one construct effective connectivity measures from subject-level time-series data that are interpretable and causally informative? Second, under what conditions is it valid to use such derived measures as outcomes for causal inference?

\subsection{Background and previous studies}

A key challenge in causal inference with multivariate outcomes is separating intervention effects from dependence among outcome components. Apparent effect on individual components may arise from associations induced by dependence with other components, rather than from direct causal effect of the intervention. One approach to isolating direct effects is to construct outcome components so that each represents a \textit{conditional} feature of the data-generating process, thereby isolating the contribution attributable to the intervention. In the context of time-series data, this can be achieved by defining each component conditional on the historical information from all remaining components.

When outcomes are indexed by multiple units (e.g., regions), identifying an appropriate conditioning set becomes challenging due to the high dimensionality. 
In effective connectivity literature, data-driven variable selection and dimension-reduction approaches are commonly used to construct conditioning sets~\citep{zhou2009analyzing, liao_evaluating_2010, tang2012measuring}. However, such approaches often introduce post-selection inference issues, as the resulting estimators and test statistics depend on data-adaptive model selection, violating assumptions underlying classical inferential procedures based on fixed models and hypotheses~\citep{kuchibhotla_post-selection_2022}. 
Sample-splitting methods are widely used to mitigate post-selection bias by separating model construction from inference through random partitioning of the data. While their validity is well understood under independence assumptions~\citep{tibshirani_exact_2016}, extending these guarantees to temporally autocorrelated data requires additional conditions~\citep{lunde_sample_2019, xi_post-selection_2022}.

Even when multivariate outcomes are constructed to be interpretable for downstream causal inference, it remains unclear under what conditions such quantities can be treated as observed outcomes for the purposes of causal inference. 
This issue has become increasingly salient with the growing use of algorithmically derived or predicted outcomes, including latent representations obtained using \textit{Artificial Intelligence} (AI) and deep learning,  in which outputs of prediction algorithms are directly incorporated as outcomes in causal effect estimation~\citep{modarressi2025causal}.  
Recent work has begun to address causal inference with multidimensional derived outcomes. \citet{qiu_unveiling_2024} establish identification conditions and multiple testing procedures for such settings, and \citet{du2025causal} extend this framework by introducing doubly robust estimation. These approaches, however, do not address dependence induced during the construction of the outcome measures themselves; rather, they account for dependence only at the level of estimated causal effects through multiple testing.

\subsection{Contribution and outline of the paper}
Our main contributions of this paper are as follows. 
First, we propose a unified approach that integrates the construction of multivariate, interdependent outcomes from time-series data with principled causal inference methods. We introduce two conceptual perspectives on derived outcomes and establish the distinct statistical and causal assumptions required for each perspective. 
Second, we propose an asymptotically valid sample-splitting strategy for constructing causally informative relational outcome measures from time-series data for valid downstream causal analysis of intervention effects. 
Third, we establish the large-sample properties of the proposed intervention effect estimators in settings with interdependent and unobservable outcomes, and show how existing multiple testing procedures can be incorporated to account for dependence among the estimators. 
Overall, these developments enable causal inference when the outcomes are derived or estimated and exhibit both within-outcome and model-derived dependencies.

This reminder of the paper is organized as follows. In Section~\ref{section2: notation}, we introduce notation, data structure, causal estimands, and identification conditions. Section \ref{section3: layer1} proposes effective connectivity measures and introduces their construction from time-series data using sample-splitting approaches. In Section \ref{section4: layer2}, we present procedures for estimation and hypothesis testing of intervention effects on the derived effective connectivity outcomes. Section~\ref{section5: simulation} conducts simulation studies evaluating the finite-sample performance of the proposed method and comparing them with alternative approaches. We apply our proposed methods to investigate the effect of amyloid PET positivity on effective connectivity in the human brain in Section~\ref{section6: application}. Section~\ref{section7:discussion} concludes with a discussion of the findings, limitations, and directions for future work.

\section{Notation and target estimands}\label{section2: notation}

\subsection{Notation and data structure}

Let $i$ ($i \in \{1,2,\ldots, n\} \coloneqq [1:n]$) index \textit{subjects} to which the unique intervention of interest can be exposed. Let $j$ ($j \in [1:p]$) index outcome \textit{units} (e.g., brain regions), with pairs of regions defining units of effective connectivity. We consider a binary intervention for simplicity, denoted by $Z_i \in \{0,1\}$ for subject $i$. To adjust for potential confounding, we consider a vector of baseline covariates $\mathbf{W}_i$ for subject $i$, assumed to be time invariant. The outcome of interest is multivariate and indexed by ordered pairs of outcome units.  Specifically, let $\mathbf{Y}_i=(Y_{i (1,2)}, Y_{i (1,3)}, ..., Y_{i (p,p-1)})^\top$, where, for example, $Y_{i (j_1 j_2)}$ represents the effective connectivity from region $j_1$ to region $j_2$ ($j_1, j_2 \in [1:p],~j_1 \neq j_2$). These outcomes are typically not directly observed and must instead be derived from subject-specific time-series data. We denote the derived outcome vector by $\widehat{\mathbf{Y}}_{i}$.
Let $X_{ij}^{(t)}$ denote the time-series measurement for subject $i$ at outcome unit $j$ at time $t$ ($i \in [1:n],~j \in [1:p],~t \in [1:T]$). Define $\mathbf{X}_{ij}^{[1:T]}$, or equivalently  $\widetilde{\mathbf{X}}_{ij}$, to represent the collection of time-series data for outcome unit $j$ of  subject $i$, and
let $\widetilde{\mathbf{X}}_i =\{\mathbf{X}_{ij}^{[1:T]}: j \in [1:p]\}$ denote the collection of time series across all outcome units for subject $i$.
In the motivating application, for example, $Y_{i (j_1 j_2)}$ corresponds to the true underlying effective connectivity from brain region $j_1$ to $j_2$ for subject $i$, and $X_{ij}^{(t)}$ represents the mean BOLD signal at brain region $j$ and time $t$, obtained from resting-state fMRI.

We can consider two distinct perspectives on the derived outcome $\widehat{\mathbf{Y}}_{i}$. Under the first perspective, $\widehat{\mathbf{Y}}_{i}$ represents an estimate of $\mathbf{Y}_{i}$, where $\mathbf{Y}_{i}$ is treated as a subject-level unknown latent parameter. In this view, causal effect estimation on $\mathbf{Y}_{i}$ remains the inferential goal, given that it is (asymptotically) equivalent to causal effect estimation on $\widehat{\mathbf{Y}}_{i}$ under suitable regularity conditions. 
Under the second perspective, $\widehat{\mathbf{Y}}_{i}$ is interpreted as a proxy for the unobservable outcome $\mathbf{Y}_{i}$. In our settings, $\mathbf{Y}_{i}$ is represented through the observed time series $\widetilde{\mathbf{X}}_i$. In this case, inference focuses on testing, rather than estimating, the causal effect on $\mathbf{Y}_{i}$ under appropriate causal assumptions.

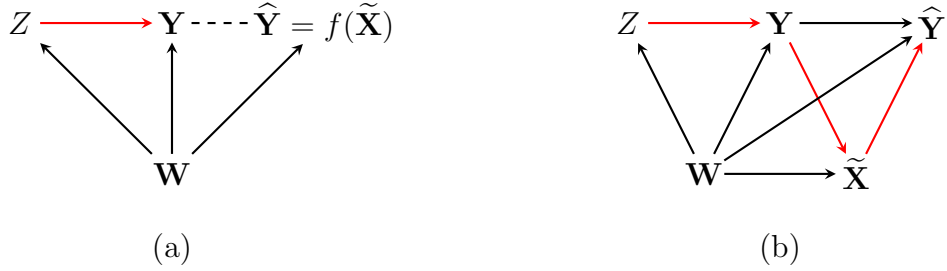
\begin{figure}[!ht]
 \begin{centering}
    \begin{tikzpicture}[>=stealth, node distance=2cm]
        \tikzstyle{format} = [thin, rectangle, minimum size=5.0mm,
        inner sep=0.1pt]
        \tikzstyle{square} = [thin, rectangle, draw]
        \begin{scope}[xshift=-4.0cm, yshift = 0cm]
            \path[->, thick]
            node[format, rectangle, line width=0.3mm] (Z) {$Z$}
            node[format, right of=Z] (Y) {$\mathbf{Y}$}
            node[format, right of=Y, line width=0.3mm] (Yhat) {$\widehat{\mathbf{Y}} = f(\widetilde{\mathbf{X}})$}
            node[format, below of=Y, rectangle, line width=0.3mm] (W) {$\mathbf{W}$}
            node[format, below of=W, rectangle, line width=0.3mm, yshift = 1cm] (a) {(a)}
                
            (Z) edge[red] (Y)
            (Y) edge[-, black, dashed] (Yhat)

            (W) edge[black] (Z)
            (W) edge[black] (Y)
            (W) edge[black] (Yhat)
            ;
        \end{scope}

           \begin{scope}[xshift=4.0cm, yshift = 0cm]
            \path[->, thick]
            node[format, rectangle, line width=0.3mm] (Z) {$Z$}
            node[format, right of=Z] (Y) {$\mathbf{Y}$}
            node[format, right of=Y] (Yhat) {$\widehat{\mathbf{Y}}$}
            node[format, below of=Y, xshift = 1cm, rectangle, line width=0.3mm] (X) {$\widetilde{\mathbf{X}}$}
            node[format, below of=Y, rectangle, xshift = -1cm, line width=0.3mm] (W) {$\mathbf{W}$}
            node[format, below of=W, rectangle, xshift = 1cm, yshift=1cm, line width=0.3mm] (b) {(b)}
                
            (Z) edge[red] (Y)
            (Y) edge[black] (Yhat)
            (X) edge[red] (Yhat)
            (Y) edge[red] (X)
            (W) edge[black] (Z)
            (W) edge[black] (Y)
            (W) edge[black] (Yhat)
            (W) edge[black] (X)
            ;
        \end{scope}
    \end{tikzpicture} 
    \par\end{centering}
    \caption{\label{fig:DAG2} A graphical representation of two approaches for using $\widehat{\mathbf{Y}}$ in causal inference on $\mathbf{Y}$. In panel (a), $\widehat{\mathbf{Y}}$ is a function of $\widetilde{\mathbf{X}}$ that is closely related to the true, unobserved outcome $\mathbf{Y}$ (indicated by the dashed line); when a standard measure of $\mathbf{Y}$ exists, the effect on a unit change in $\mathbf{Y}$ can be estimated by the effect on a unit change in $\widehat{\mathbf{Y}}$. In panel (b), $\widehat{\mathbf{Y}}$ serves as a proxy for $\mathbf{Y}$, and the units of $\mathbf{Y}$ and $\widehat{\mathbf{Y}}$ are not necessarily comparable; the intervention $Z$ is assumed to affect $\widehat{\mathbf{Y}}$ through $\mathbf{Y}$ and the intermediate outcome $\widetilde{\mathbf{X}}$.}
\end{figure}

Figure~\ref{fig:DAG2} illustrates two different perspectives on using $\widehat{\mathbf{Y}}$ for causal inference on $\mathbf{Y}$. In Figure~\ref{fig:DAG2}(a), $\widehat{\mathbf{Y}}$ is defined as a function of $\widetilde{\mathbf{X}}$ and is closely related to the true, unobserved outcome $\mathbf{Y}$ (denoted by the dashed line), while allowing for confounding between $\mathbf{Y}$ and $\widehat{\mathbf{Y}}$ through observed baseline covariates $\mathbf{W}$. This perspective is appropriate when a specific, and presumably unique, measure of the true outcome $\mathbf{Y}$ exists, such that the causal effect corresponding to a unit change in $\mathbf{Y}$ can be approximated by the effect on a unit change in $\widehat{\mathbf{Y}}$. 
Figure~\ref{fig:DAG2}(b) illustrates the second perspective, in which $\widehat{\mathbf{Y}}$ is treated as a proxy for the unobservable outcome $\mathbf{Y}$. This perspective is appropriate when the target outcome $\mathbf{Y}$ is an abstract concept that may be measured in multiple ways, but for which the intervention is expected to affect the proxy $\widehat{\mathbf{Y}}$ through $\mathbf{Y}$ and the intermediate outcome $\widetilde{\mathbf{X}}$. Although $\mathbf{Y}$ may influence $\widehat{\mathbf{Y}}$ through multiple pathways, the focus here is on the pathway mediated by the observed time series $\widetilde{\mathbf{X}}$. In this setting, the units of $\mathbf{Y}$ (if such units exist) and $\widehat{\mathbf{Y}}$ are not necessarily comparable. Under this perspective, inference is directed toward testing the existence of a causal effect of $Z$ on $\mathbf{Y}$ through testing for an effect of $Z$ on $\widehat{\mathbf{Y}}$.
We elaborate on the conditions required for each perspective in Section~\ref{ssec:identification}.

\subsection{Potential outcomes and causal estimands}

We use the potential outcomes framework to define the target estimands and formalize the causal null hypothesis of interest~\citep{neyman1923application, rubin1974estimating}. For each outcome component, let $Y_{i (j_1 j_2)}(z)$ denote the potential outcome under intervention level $z \in \{0,1\}$. The vector $\mathbf{Y}_i(z)=(Y_{i (1,2)}(z),..., Y_{i (p,p-1)}(z))^\top$ includes potential outcomes for subject $i$ across ordered pairs of outcome units. We examine the difference between the potential outcomes $\mathbf{Y}_i(0)$ and $\mathbf{Y}_i(1)$ using post-intervention time-series data $\widetilde{\mathbf{X}}_i$ and the estimated or derived outcomes $\widehat{\mathbf{Y}}_i$. We introduce analogous potential outcome notation for the observed time series and derived outcomes, denoting them by $\widetilde{\mathbf{X}}_i(z)$ and $\widehat{\mathbf{Y}}_i(z)$, respectively, for $z \in \{0, 1\}$. 

Our focus is on the causal effect of the intervention on individual components of the outcome $Y_{(j_{1} j_{2})}~(j_1, j_2  \in [1:p];~j_1 \neq j_2$). These effects can be represented by differences in average potential outcomes,
\begin{eqnarray} \label{eq:tau}
\tau_{(j_{1} j_{2})} = \mathbb{E}_{p}[Y_{ (j_{1} j_{2})} (1) -  Y_{(j_{1} j_{2})} (0)],
\end{eqnarray}
where the expectation $\mathbb{E}_{p}(\cdot)$ is taken over subjects in the population. Let $\boldsymbol{\tau} = (\tau_{(1,2)},...,\tau_{(p,p-1)})^\top= \mathbb{E}_{p}[\mathbf{Y}(1) - \mathbf{Y}(0)]$ denote the vector of componentwise causal effects. Because the true outcomes $\mathbf{Y}_{i}$ are not directly observed,  inference under either of the two perspectives described earlier directly targets $\boldsymbol{\tau}^{*} : = \mathbb{E}_{p}[\widehat{\mathbf{Y}}(1) - \widehat{\mathbf{Y}}(0)]$, which is the corresponding causal estimand defined in terms of the derived outcomes. 

Under the second perspective on $\widehat{\mathbf{Y}}_{i}$, the inferential focus shifts from estimation of the causal effect in~\eqref{eq:tau} to hypothesis testing. Specifically, for $j_1, j_2 \in [1:p];~j_1 \neq j_2$, we consider the collection of null hypotheses
\begin{eqnarray} \label{eq:null}
 H_{0, (j_1 j_2)}: \tau_{(j_1 j_2)}=0.
\end{eqnarray}
\noindent Similarly, for the derived outcomes we define the corresponding class of null hypotheses $\{ H^{*}_{0, (j_1 j_2)}: \tau^{*}_{(j_1 j_2)}=0$; $j_1, j_2 \in [1:p];~j_1 \neq j_2 \}$.

\subsection{Identification}
\label{ssec:identification}

In this section, we first establish the causal assumptions required under each perspective and distinguish those that are common from those that are perspective specific. For $z \in \{0,1\}$, $i \in [1:n]$, we impose the following common assumptions.

\begin{assumption}[Consistency] \label{as:assumption1} If $Z_i = z \in \{0,1\}$,  
$\widetilde{\mathbf{X}}_i = \widetilde{\mathbf{X}}_i(z),
\mathbf{Y}_i = \mathbf{Y}_i(z),
\widehat{\mathbf{Y}}_i = \widehat{\mathbf{Y}}_i(z)$.
\end{assumption}

\begin{assumption}[Stable Unit Treatment Value Assumption (SUTVA)~\citep{rubin1974estimating}] \label{as:assumption2}
    There are no hidden variations in the intervention and no interference between subjects.
\end{assumption}

\begin{assumption}[Positivity] \label{as:assumption3} $Pr(Z_i=z \mid \mathbf{W}_i = \boldsymbol{w})>0$ for $\boldsymbol{w} \in \Omega_{W}$.
\end{assumption}

\begin{assumption}[Ignorability] \label{as:assumption4}
$Z \indep \widetilde{\mathbf{X}}(0), \widetilde{\mathbf{X}}(1) \mid \mathbf{W}$.
\end{assumption}
\noindent Consistency ensures that the observed time series, the unobservable outcomes, and the derived outcomes coincide with their respective potential outcome counterparts. The SUTVA assumes a single version of the intervention and no interference between subjects. The positivity and ignorability assumptions require sufficient overlap in the propensity scores and unconfoundedness of the intervention, conditional on the baseline covariates $\mathbf{W}_i$. Under either perspective, $\widehat{\mathbf{Y}}_i(z)$ is a function of $\widetilde{\mathbf{X}}_i(z)$ and $\mathbf{W}_i$. Therefore, Assumption~\ref{as:assumption4} implies unconfoundedness of the intervention with respect to the potential derived outcomes, that is, $Z \indep \widehat{\mathbf{Y}}(0), \widehat{\mathbf{Y}}(1) \mid \mathbf{W}$.

Because the outcomes $\mathbf{Y}_{i}$ are not directly observed, additional conditions are required to relate $\{\mathbf{Y}_i(1), \mathbf{Y}_i(0)\}$ to $\{\widehat{\mathbf{Y}}_i(1), \widehat{\mathbf{Y}}_i(0) \}$. Under the first perspective, where we treat $\widehat{\mathbf{Y}}_{i}$ as an estimator of $\mathbf{Y}_{i}$ (equivalently, $\widehat{\mathbf{Y}}_{i}(z)$ for $\mathbf{Y}_{i}(z)$ for $Z_{i} = z$ under consistency). Then we define the bias in the subject-level outcome measures by:
\[
\boldsymbol{\Delta}_{i}(z)=(\Delta_{i(1,2)}(z), ..., \Delta_{i(p,p-1)}(z))^\top=\mathbb{E}_{i}[\widehat{\mathbf{Y}}_i(z) \mid \widetilde{\mathbf{X}}_i(z), \mathbf{W}_i] - \mathbf{Y}_{i}(z),
\]
where $\mathbb{E}_{i}[\cdot \mid \widetilde{\mathbf{X}}_i(z), \mathbf{W}_i]$ denote the subject-specific expectation given $\widetilde{\mathbf{X}}_{i}(z)$ and $\mathbf{W}_i$. The following condition  restricts the asymptotic behavior of the bias in $\widehat{\mathbf{Y}}_i(z)$ as the time-series dimension $T$ increases. 

\begin{condition}[Asymptotic Unbiasedness] \label{con:condition1}
For $z \in \{0,1\}$, $\max \limits_{j_1, j_2=[1:p];~j_1 \neq j_2} \mathbb{E}_{p} \{\Delta_{(j_1 j_2)}(z)\}=o_{p}(1)$, as $ T \rightarrow \infty$.
\end{condition}
\noindent 
Condition~\ref{con:condition1} allows us to treat $\widehat{\mathbf{Y}}_{i}$ as an estimator of $\mathbf{Y}_{i}$ and is a key requirement for the identification result stated in the following theorem.
\begin{theorem} \label{thm:theorem1}
    Under Assumptions \ref{as:assumption1}--\ref{as:assumption4} and Condition \ref{con:condition1}, the estimand $\boldsymbol{\tau}^{*}$ is consistent for $\boldsymbol{\tau}$ as $T \rightarrow \infty$.
\end{theorem}

On the other hand, when $\widehat{\mathbf{Y}}_{i}$ is interpreted as a proxy for $\mathbf{Y}_{i}$ rather than an estimator, the following assumptions ensure that hypothesis tests based on the derived outcomes remain valid.

\begin{Aassumption}[Relevance] \label{as:assumptionA1}
     The unobservable outcomes of interest are associated with the derived outcomes.
\end{Aassumption}

\begin{Aassumption}[Exclusion Restriction] \label{as:assumptionA2}
    The intervention can affect the derived outcomes $\widehat{\mathbf{Y}}_i$, only through the unobservable outcome $\mathbf{Y}_i$. 
\end{Aassumption}

\begin{Aassumption}[Ignorability]
    $Z \indep \mathbf{Y}(0), \mathbf{Y}(1) \mid \mathbf{W}$.
\label{as:assumptionA3}

\end{Aassumption}

\noindent Relevance ensures that any intervention on unobservable outcomes is reflected in the derived outcomes. The exclusion restriction requires that that all effects of the intervention on $\widehat{\mathbf{Y}}$ operate through $\mathbf{Y}$, ruling out direct pathways. This assumption is conceptually related to the exclusion restriction assumptions in identifying latent treatment effects~\citep{stoetzer2025causal} and in instrumental variable models, although in our setting $\widehat{\mathbf{Y}}$ serves as a proxy for the intermediate unobservable $\mathbf{Y}$. Finally, the ignorability assumption requires that, conditional on $\mathbf{W}$, the intervention is unconfounded with the potential unobservable outcomes.

\begin{theorem} \label{thm:theorem2}
    Under Assumptions \ref{as:assumption1}--\ref{as:assumption4}, \ref{as:assumptionA1}--\ref{as:assumptionA3}, a test of $H^{*}_{0, (j_{1} j_{2})}: \tau^{*}_{(j_{1} j_{2})} = 0$ provides a valid test of $H_{0, (j_{1} j_{2})}: \tau_{(j_{1} j_{2})} = 0$ for $j_{1}, j_{2} \in [1:p];~j_{1} \neq j_{2}$.
\end{theorem}

\noindent Given these assumptions, Section~\ref{section3: layer1} focuses on obtaining subject-level estimates $\widehat{\mathbf{Y}}_{i}$ from the time-series data $\widetilde{\mathbf{X}}_{i}$ for each subject $i$. Section~\ref{section4: layer2} then proceeds to across-subject inference to examine causal effects by combining data from both the intervention and control groups.

\section{Deriving multivariate interdependent outcomes} \label{section3: layer1}

To distinguish changes in each outcome component that are directly due to the external intervention from those induced indirectly through its effects on other components, we construct a derived outcome measure, denoted by $\widehat{\mathbf{Y}}_{i}$, from the observed time-series data $\widetilde{\mathbf{X}}_{i}$. 
In this section, we omit the subscript $i$, the subject index, for simplicity.

\subsection{Granger causality}

A time series $\widetilde{\mathbf{X}}_{\cdot j}$ is said to \textit{Granger-cause} another series $\widetilde{\mathbf{X}}_{\cdot k}$ if past values of $\widetilde{\mathbf{X}}_{\cdot j}$ contain information that improves prediction of the future values of $\widetilde{\mathbf{X}}_{\cdot k}$~\citep{granger_investigating_1969}. Although Granger causality reflects predictive relationships rather than true causation, we use the estimated Granger causality as a derived outcome measure for brain effective connectivity, given their widespread use in neuroscience. To account for the influence of other brain regions that may confound the relationship between regions $j$ and $k$, we employ \textit{conditional} Granger causality, which conditions on the time series of additional variables~\citep{geweke_measures_1984, liao_evaluating_2010}. This approach allows the derived outcomes to better capture directed dependencies between regions while accounting for the broader multivariate system. 

Figure~\ref{fig:DAG3} illustrates the role of conditional Granger causality when examining the causal effect of an external intervention on pairwise effective connectivity. The figure shows a case where the effective connectivity measure for one pair, if not carefully constructed, may be influenced by the intervention’s effects on other pairs in a \textit{non-causal} way. In Figure~\ref{fig:DAG3}(a), under control (i.e., $Z=0$), the time series $\widetilde{\mathbf{X}}_{\cdot 3}$ Granger-causes $\widetilde{\mathbf{X}}_{\cdot 1}$. In Figure~\ref{fig:DAG3}(b), under intervention (i.e., $Z=1$), the strength of this existing Granger causality from $\widetilde{\mathbf{X}}_{\cdot 3}$ to $\widetilde{\mathbf{X}}_{\cdot 1}$ increases, and an additional Granger-causal effect from $\widetilde{\mathbf{X}}_{\cdot 3}$ to $\widetilde{\mathbf{X}}_{\cdot 1}$ emerges.
When Granger causality is used as the measure of effective connectivity, such changes would lead to detected effects on $Y_{(3,1)}$ and $Y_{(3,2)}$. One might also detect apparent effects on $Y_{(1,2)}$ or $Y_{(2,1)}$, even though there is no change in the Granger-causality relationship between  $\widetilde{\mathbf{X}}_{\cdot 1}$ and $\widetilde{\mathbf{X}}_{\cdot 2}$, simply because their relationships with $\widetilde{\mathbf{X}}_{\cdot 3}$ have changed. 
This example motivates defining $\widehat{\mathbf{Y}}_{(j,k)}$ using conditional Granger causality, so that the resulting measure for pair $(j,k)$ is not affected by the time series $\widetilde{\mathbf{X}}_{\cdot h}$ for any $h$ ($h \neq j,k$).  

\begin{figure}
\centering
\begin{tikzpicture}[>=stealth, node distance=1cm]
  \tikzstyle{format} = [thin, rectangle, minimum size=5.0mm, inner sep=0.1pt]
  \tikzstyle{square} = [thin, rectangle, draw]

  \begin{scope}[xshift=-4.0cm, yshift=0cm]
    \node[format] (z) {$Z=0$};

    \node[right=of z, draw, inner sep=3pt] (x) {%
      \begin{tikzpicture}[scale=0.8]
        \node[format] (x3) at (0,0) {$\widetilde{\mathbf{X}}_{\cdot 3}(0)$};
        \node[format] (x1) at (-1, -1.5) {$\widetilde{\mathbf{X}}_{\cdot 1}(0)$};
        \node[format] (x2) at ( 1, -1.5) {$\widetilde{\mathbf{X}}_{\cdot 2}(0)$};
        \draw[->, ForestGreen, thick] (x3) -- (x1);

      \end{tikzpicture}
    };

    \node[format, below of=z, rectangle, line width=0.3mm, xshift = 1cm, yshift = -1cm] (a) {(a)};

  \end{scope}

  \begin{scope}[xshift=4.0cm, yshift=0cm]
    \node[format] (z) {$Z=1$};

    \node[right=of z, draw, inner sep=3pt] (x) {%
      \begin{tikzpicture}[scale=0.8]
        \node[format] (x3) at (0,0) {$\widetilde{\mathbf{X}}_{\cdot 3}(1)$};
        \node[format] (x1) at (-1, -1.5) {$\widetilde{\mathbf{X}}_{\cdot 1}(1)$};
        \node[format] (x2) at ( 1, -1.5) {$\widetilde{\mathbf{X}}_{\cdot 2}(1)$};
        \draw[->, ForestGreen, thick, line width=0.75mm] (x3) -- (x1);
        \draw[->, ForestGreen, thick] (x3) -- (x2);
      \end{tikzpicture}
    };

    \node[format, below of=z, rectangle, line width=0.3mm, xshift = 1cm, yshift = -1cm] (b) {(b)};
  
  \end{scope}
\end{tikzpicture}
\caption{\label{fig:DAG3} Representations of Granger causality relationships under (a) control and (b) intervention, where the presence of a direct arrow indicates Granger causality between two time series and the arrow width reflects the strength of that relationship.}
\end{figure}
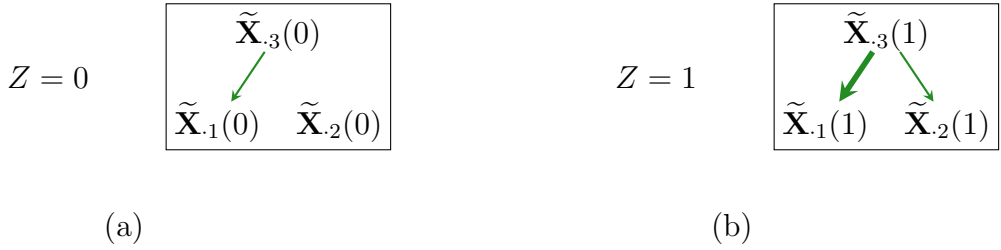

\begin{remark}
The setting illustrated in Figure~\ref{fig:DAG3}(b) differs from a standard causal mediation framework (e.g., a causal pathway 
$Z \rightarrow Y_{(1,3)} \rightarrow Y_{(1,2)}$) in that the spurious association 
between $\widetilde{\mathbf{X}}_{\cdot 1}$ and $\widetilde{\mathbf{X}}_{\cdot 2}$ 
is not necessarily causally induced. Accordingly, the notion of ``direct'' effect in this work does not correspond to the effect obtained by blocking an indirect causal path between $Z$ and $Y_{(1,2)}$. Instead, it refers to the effect obtained by blocking any paths between $Z$ and $Y_{(1,2)}$ that operate through other components $Y_{(j, k)}$ with $(j, k) \neq (1,2)$.
\end{remark}

\subsection{Deriving conditional Granger causality}\label{subsec: condGranger}

To derive conditional Granger causality, we employ vector autoregressive (VAR) models with $\widetilde{\mathbf{X}}$ as input, under the assumption of weak stationarity; that is, the mean and covariance matrix are time-invariant. 
To estimate the conditional Granger causality from $\widetilde{\mathbf{X}}_{\cdot j_{1}}$ to $\widetilde{\mathbf{X}}_{\cdot j_{2}}$, we construct two VAR models of order $r$ for $\widetilde{\mathbf{\mathbf{X}}}_{\cdot j_{2}}$. The full model includes lagged values of $\widetilde{\mathbf{X}}_{\cdot j_{1}}$ as predictors, whereas the reduced model excludes them. To account for the influence of other time series, both models may additionally adjust for lagged values of $\widetilde{\mathbf{X}}_{\cdot k}$ for $k \in \mathcal{J}$, where $\mathcal{J} \subseteq [1:p] \setminus \{j_{1}, j_{2}\}$. We denote this collection by $\widetilde{\mathbf{X}}_{\cdot \mathcal{J}}$, and let $\ell = |\mathcal{J}|$ be the number of conditioned time series. The models may also incorporate baseline covariates $\mathbf{W}$ to account for measurement error or other sources of systematic variation.

The coefficients in the VAR models are estimated by ordinary least squares under the assumption that the error terms are multivariate normal with mean zero and a time-invariant covariance matrix.  
Using the estimated coefficients from the full and reduced models, we construct the following test statistic for the null hypothesis that $\widetilde{\mathbf{X}}_{\cdot j_1}$ does not Granger-cause $\widetilde{\mathbf{X}}_{\cdot j_2}$ conditioning on $\widetilde{\mathbf{X}}_{\cdot \mathcal{J}}$.
\begin{equation} \label{eq:Fhat}
    \widehat F_{(j_1 j_2)} = \frac{\sum_{t=r+1}^{T} (\widehat{\varepsilon}_{j_2}^{(t)})^2-\sum_{t=r+1}^{T} (\widehat{\varepsilon}'^{(t)}_{j_2})^2}{r} \left\{ \frac{\sum_{t=r+1}^{T}(\widehat{\varepsilon}'^{(t)}_{j_2})^2}{(T-r)-(\ell+2)r-1} \right\}^{-1},
\end{equation}
\noindent where $\widehat{\varepsilon}_{j_2}^{(t)}$ and $\widehat{\varepsilon}'^{(t)}_{j_2}$ denote the sample residuals from modeling $X_{\cdot j_2}^{(t)}$ in the reduced and full models, respectively. The test statistic asymptotically follows $F$-distribution with degrees of freedom $(r,(T-r)-(\ell+2)r-1)$.  Although the residuals are indexed by $t$, their variances are assumed to be time-invariant.

From the first perspective, illustrated in Figure~\ref{fig:DAG2}(a), the unobserved outcome $\mathbf{Y}$ may be viewed as the collection of $\{F_{(j_{1} j_{2})}: j_{1}, j_{2} \in [1:p],~j_{1} \neq j_{2}\}$, where $F_{(j_1 j_2)}$ denotes the population analogue of $\widehat{F}_{(j_{1} j_{2})}$ obtained by replacing the sample residual sums of squares with their true population quantities. Then $\{\widehat{F}_{(j_{1} j_{2})}: j_{1}, j_{2} \in [1:p],~j_{1} \neq j_{2}\}$ represents the estimator $\widehat{\mathbf{Y}}$ for each subject $i$. 
From the second perspective, shown in Figure~\ref{fig:DAG2}(b), the unobserved outcome $\mathbf{Y}$ implies the underlying effective connectivity, while $\{\widehat{F}_{(j_{1} j_{2})}: j_{1}, j_{2} \in [1:p],~j_{1} \neq j_{2}\}$ serves as the derived outcome vector $\widehat{\mathbf{Y}}$ for each subject $i$.

\subsection{Sample-splitting approach}

In our setting with moderately large $p$ and $T$, the conditioning sets (e.g., $\mathcal{J}$) and the VAR order (e.g., $r$) may be chosen using a data-driven procedure for each subject. 
These choices affect both the null hypothesis and the null distribution of the statistic $\widehat{F}_{(j_{1}j_{2})}$ in~\eqref{eq:Fhat}. If $\widehat{F}_{(j_{1}j_{2})}$ is constructed after selecting $\mathcal{J}$ and $r$ using the same data, a selective inference problem arises, leading to anti-conservative inference.

To address this issue, we adopt a sample-splitting approach applied to time-series data for constructing $\widehat{F}_{(j_{1}j_{2})}$ for each subject.
For simplicity, suppose that each subject's time series is divided into two equal parts: the first half and the second half, with $T_0 = T/2$.
Let $\mathcal{D}_1$ denote the first part, $\mathcal{D}_{1} := \{X_{\cdot j}^{(t)}: j = [1:p], ~t= [1:T_0]\}$, which can be used for model selection (i.e., selecting $\mathcal{J}$ and $r$).
Let $\mathcal{D}_{2}$ denote the second part, $\mathcal{D}_{2} := \{ X_{\cdot j}^{(t)}: j =[1:p], ~t=[(T_0+1):T]\}$, which can be used to compute $\widehat{F}_{(j_{1}j_{2})}$ in~\eqref{eq:Fhat}.  Unlike the i.i.d. setting, time-series data must be split into contiguous time blocks rather than randomly sampled time indices, so that the serial dependence structure learned from $\mathcal{D}_{1}$ can be validly applied to $\mathcal{D}_{2}$. 

Let $\widehat{F}_{(j_1 j_2)}(\mathcal{D}_2) \mid \widehat{m}(\mathcal{D}_1)$ denote the statistic in~\eqref{eq:Fhat}, computed using $\mathcal{D}_{2}$ under the model specification $\widehat{m}(\mathcal{D}_1)$ selected from $\mathcal{D}_{1}$. 
Next, we establish the \textit{asymptotic} validity of this sample-splitting procedure for $\widehat{F}_{(j_1 j_2)}(\mathcal{D}_2) \mid \widehat{m}(\mathcal{D}_1)$, even when $\mathcal{D}_{1}$ and $\mathcal{D}_{2}$ are statistially dependent.

\begin{definition}[Asymptotic validity of sample-splitting~\citep{lunde_sample_2019}] \label{df:definition1}
Let $\widetilde{\mathcal{D}}_2$ denote an independent copy of $\mathcal{D}_2$ that is also independent of $\mathcal{D}_1$. 
The sample-splitting is said to be asymptotically valid in probability under a metric $d(\cdot, \cdot)$ measuring the distance between conditional distributions (e.g., bounded Lipschitz metric) if
$$
d \left( \mathcal{L}(\widehat{F}_{(j_1 j_2)}(\mathcal{D}_2) \mid \widehat{m}(\mathcal{D}_1)), \, \mathcal{L}(\widehat{F}_{(j_1 j_2)}(\widetilde{\mathcal{D}}_2) \mid \widehat{m}(\mathcal{D}_1)) \right) \overset{p}{\longrightarrow} 0 \text{ as } T \rightarrow \infty,
$$ 
where $\mathcal{L}(\cdot)$ denotes the distribution of a random variable. 
\end{definition}

Sufficient conditions for asymptotic validity of the sample-splitting procedure include (i) tightness and (ii) asymptotic deletion stability~\citep{lunde_sample_2019}.
For (i), under the null hypothesis, $\widehat{F}_{(j_1 j_2)}(\mathcal{D}_2) \mid \widehat{m}(\mathcal{D}_1)$ converges in distribution to $F(r,(T-T_0-r)-(\ell+2)r-1)$ as $T \rightarrow \infty$.
Convergence to a non-degenerate limit implies asymptotic tightness; that is, for any $\delta > 0 $, there exists a compact interval such that the probability that $\widehat{F}_{(j_1 j_2)}(\mathcal{D}_2) \mid \widehat{m}(\mathcal{D}_1)$ lies outside this interval is smaller than $\delta$. 

For (ii), we formalize asymptotic deletion stability for the proposed outcome measures $\widehat{F}_{(j_1 j_2)}(\mathcal{D}_2) \mid \widehat{m}(\mathcal{D}_1)$  as follows.
\begin{definition} [Asymptotic deletion stability] \label{df:definition2}
    Assume that $T-T_0 \rightarrow \infty$. Suppose that there exists a sequence $L_{T-T_0}$ such that $L_{T-T_0} \rightarrow \infty$ and  $T-T_0-L_{T-T_0} \rightarrow \infty$. Let $\mathcal{D}^{\prime}_2$ denote a version of $\mathcal{D}_2$ with its first $L_{T-T_0}$ observations removed. Consider an estimator $\hat{\theta} \mid m$ for each $m \in \mathcal{M}^{*}$, where $|\mathcal{M}^{*}| < \infty$ and $Pr(\widehat{m}(\mathcal{D}_{1}) \in \mathcal{M}^{*}) \rightarrow 1$. The estimator $\hat{\theta} \mid \widehat{m}$ is said to be \textit{asymptotically deletion stable} if
    $$\widehat{\theta}(\mathcal{D}_2) \mid \widehat{m}(\mathcal{D}_1) - \widehat{\theta}(\mathcal{D}^{\prime}_2) \mid \widehat{m}(\mathcal{D}_1) \overset{p}{\longrightarrow} 0.$$
\end{definition}
\noindent The following theorem shows that, under suitable regularity conditions, the proposed outcome measure for effective connectivity satisfies asymptotic deletion stability, ensuring that the corresponding sample-splitting approach is asymptotically valid.

\begin{theorem}
\label{thm:theorem3}
Consider the setting of Definition~\ref{df:definition2}.
Under the null hypothesis that $\widetilde{\mathbf{X}}_{\cdot j_1}$ does not Granger-cause $\widetilde{\mathbf{X}}_{\cdot j_2}$ conditioning on $\widetilde{\mathbf{X}}_{\cdot \mathcal{J}}$, suppose that 
\[
\lim\limits_{T\to\infty} \frac{(T-T_0-r)-(\ell+2)r-1}{(T-T_0-L_{T-T_0}-r)-(\ell+2)r-1} =1.
\]
Then  
\[
 \widehat{F}_{(j_1 j_2)}(\mathcal{D}_2) \mid \widehat{m}(\mathcal{D}_1) - \widehat{F}_{(j_1 j_2)}(\mathcal{D}'_2) \mid \widehat{m}(\mathcal{D}_1) \overset{p}{\longrightarrow} 0.
 \]
Together with tightness, this implies that $\widehat{F}_{(j_1 j_2)}(\mathcal{D}_2) \mid \widehat{m}(\mathcal{D}_1)$ provides an asymptotically valid sample-splitting procedure. 
\end{theorem}

The proposed sample-splitting approach is particularly useful for time-series settings in which the derived variables represent relational outcomes among interconnected units and must be learned from longitudinal data. As the weak stationarity assumption is the key to establishing asymptotic validity, the time-series data used should consist entirely of either pre-intervention or post-intervention periods, so that the series is not disrupted by the intervention.

\section{Causal inference with derived multivariate outcomes}\label{section4: layer2}

With the derived outcomes $\widehat{\mathbf{Y}}_{i}$, we now consider inference for the intervention effect on them. We focus on (i) addressing confounding due to covariates $\mathbf{W}_i$ under the ignorability assumption (Assumptions~\ref{as:assumption4} and~\ref{as:assumptionA3}), and (ii) accounting for the multiple testing problem when testing intervention effects across multiple pairs of $(j_{1}, j_{2})$ in~\eqref{eq:tau} or~\eqref{eq:null}.

\subsection{Treatment effect estimation for derived outcomes}
We first aim to estimate each component of $\boldsymbol{\tau}^{*}  = \mathbb{E}_{p}[\widehat{\mathbf{Y}}(1) - \widehat{\mathbf{Y}}(0)]$; in particular, $\tau^{*}_{(j_{1} j_{2})} = \mathbb{E}_{p}[\widehat{Y}_{(j_{1}j_{2})}(1) - \widehat{Y}_{(j_{1}j_{2})}(0)]$. 
To address confounding between $Z$ and $\widehat{\mathbf{Y}}$ under Assumption~\ref{as:assumption4}, we construct an inverse probability weighting (IPW) estimator. One could alternatively consider an augmented IPW estimator that incorporates an outcome regression model to improve robustness to model misspecification. However, with multivariate outcomes exhibiting within-variable dependence, specifying a model for $\mathbb{E}_{p}(\widehat{\mathbf{Y}} \mid Z , \mathbf{W})$ is nontrivial. We therefore focus on the IPW estimator.

With the propensity score defined as $ e(\mathbf{W}_i)  \coloneqq Pr(Z=1 \mid \mathbf{W}_i)$, the IPW estimator for $\tau^{*}_{(j_{1}j_{2})}$ is:
\begin{equation} \label{eq:IPW}
    \widehat{\tau}^{*}_{(j_{1} j_{2})}= \frac{1}{n}\sum_{i=1}^n \left\{ \frac{Z_i \widehat{Y}_{i(j_{1} j_{2})}}{\widehat{e}(\mathbf{W}_i)} - \frac{(1-Z_i)\widehat{Y}_{i(j_{1} j_{2})}}{1-\widehat{e}(\mathbf{W}_i)} \right\}.
\end{equation}
Under Assumptions \ref{as:assumption1}--\ref{as:assumption4} and a correctly specified propensity score model, it is straightforward to show that $\widehat{\boldsymbol{\tau}}^{*}$ is a consistent estimator of $\boldsymbol{\tau}^{*}$. Deriving its asymptotic variance and establishing a normal approximation requires additional care, however, because the components of $\widehat{\tau}^{*}$ are generally correlated.

Consider a parametric model for the propensity score of the form of $e(\mathbf{X}_{i}) = \sigma(\mathbf{W}_i^\top \bm{\beta})$, where $\sigma(\cdot)$ is a known smooth link function. Let $\widehat{\boldsymbol{\beta}}$ denote the maximum likelihood estimator of $\boldsymbol{\beta}$, and define $\widehat e(\mathbf W_i) = \sigma(\mathbf W_i^\top \widehat{\boldsymbol{\beta}})$. Write $\eta_i = \mathbf{W}_i^\top \bm{\beta}$, and let $\sigma_i' = \sigma'(\eta_i)$ denote the derivative of $\sigma(\eta_i)$. Under regularity conditions, $\widehat{\tau}^{*}_{(j_{1} j_{2})}$ has the following asymptotically linear representation,
\begin{eqnarray} \label{eq:expansion}
\sqrt{n}\left( \widehat{\tau}^{*}_{(j_{1} j_{2})} - \tau^{*}_{(j_{1} j_{2})}  \right) = \frac{1}{\sqrt{n}}\sum_{i=1}^{n} h_{i(j_{1} j_{2})} + o(n^{-1/2}),
\end{eqnarray}
where $h_{i(j_{1}j_{2})}$ is the influence function of $\widehat{\tau}_{(j_{1}j_{2})}^{*}$. Specifically,
\begin{align*}
h_{i(j_{1} j_{2})} & = \frac{Z_i \widehat{Y}_{i(j_{1} j_{2})}}{\sigma_i} - \frac{(1-Z_i)\widehat{Y}_{i(j_{1} j_{2})}}{1-\sigma_i} - \tau^{*}_{(j_{1} j_{2})}\\
& - \left[ \frac{1}{n}\sum_{i=1}^n 
\left[Z_i\widehat{Y}_{i(j_{1} j_{2})}\,\frac{\sigma_i'}{\sigma_i^2}
+ (1-Z_i)\widehat{Y}_{i(j_{1} j_{2})}\,\frac{\sigma_i'}{(1-\sigma_i)^2}
\right]\mathbf W_i^\top \right] \mathcal I_{\bm{\beta}}^{-1}
(Z_i - \sigma_i) \frac{\sigma_i'}{\sigma_i(1-\sigma_i)} \mathbf W_i,
\end{align*}
where $\mathcal I_{\bm{\beta}}$ denotes the Fisher information matrix for $\bm{\beta}$. 

Compared to the derivations in~\cite{qiu_unveiling_2024}, in which the proposed IPW estimator based on the derived multivariate outcomes $\widehat{\mathbf{Y}}$ targets the estimand defined in terms of the potential outcomes of the unobserved outcome $\mathbf{Y}$ (i.e., $\boldsymbol{\tau}$), our inference targets instead the estimand defined in terms of potential outcomes of the derived outcomes $\widehat{\mathbf{Y}}_{i}$ (i.e., $\boldsymbol{\tau}^{\ast}$). Moreover, the link function $\sigma(\cdot)$ can be any known smooth mapping from $\mathbf{W}_i^\top \bm{\beta}$ to $(0,1)$, and need not be restricted to the logistic link. Under such general setting, we establish the asymptotic normality of $\widehat{\boldsymbol{\tau}}^{\ast}$ in the following theorem.

\begin{theorem}\label{thm:theorem4}
    Suppose that the propensity score model is correctly specified and Assumptions \ref{as:assumption1}--\ref{as:assumption4} hold. If $\mathcal{V}_{(j_{1} j_{2})}:= \mathbb{E}(\frac{1}{n}\sum_{i=1}^{n} h_{i(j_{1} j_{2})}^2) < \infty$, then 
    \[ \sqrt{n}(\widehat{\tau}^{*}_{(j_{1} j_{2})} - \tau^{*}_{(j_{1} j_{2})})  \overset{d}{\longrightarrow} N(0,\mathcal{V}_{(j_{1} j_{2})}) \text{ as } n \rightarrow \infty.
    \]
\end{theorem}
\noindent An analytic estimator of the asymptotic variance, $\widehat{\mathcal{V}}_{(j_{1} j_{2})}$, can be derived by $\frac{1}{n}\sum_{i=1}^{n} \widehat{h}_{i(j_{1}  j_{2})}^2$, where $\widehat{h}_{i(j_{1} j_{2})}$ is obtained by replacing $\boldsymbol{\beta}$ with $\widehat{\boldsymbol{\beta}}$ in $h_{i(j_{1} j_{2})}$.

\begin{corollary}\label{corollary1}
    Under the conditions of Theorem~\ref{thm:theorem4}, $\widehat{\mathcal{V}}_{(j_1 j_2)}$ is a consistent estimator of $\mathcal{V}_{(j_{1}j_{2})}$. 
\end{corollary}
\noindent In practice, we restrict attention to components $(j_{1}, j_{2})$ such that $\widehat{\mathcal{V}}_{(j_{1} j_{2})} > c$ for some 
small constant $c > 0$, thereby excluding components associated with degenerate limiting distributions.

Under the additional condition below, the estimator $\widehat{\tau}_{(j_{1}j_{2})}$ may be interchangeably with $\widehat{\tau}^{*}_{(j_{1}j_{2})}$ when the derived outcomes are asymptotically unbiased for the unobserved outcomes (Condition~\ref{con:condition1}). 
\begin{condition} \label{con:condition2}
For $z \in \{0,1\}$, as $T,~n \rightarrow \infty$,  $\max \limits_{j_1, j_2=[1:p]; ~j_1 \neq j_2} \mathbb{E}_{p} \left[ \Delta_{(j_{1} j_{2})}(z) \right]=o_{p}(n^{-1/2})$, where $\Delta_{i(j_{1}j_{2})}(z)=\mathbb{E}_{i}[\widehat{Y}_{i(j_{1}j_{2})}(z) \mid \widetilde{\mathbf{X}}_i(z), \mathbf{W}_i] - Y_{i(j_{1}j_{2})}(z)$.
\end{condition}
\noindent Under Condition~\ref{con:condition2}, $\sqrt{n} (\boldsymbol{\tau}^{\ast} - \boldsymbol{\tau}) \overset{p}{\to} 0$ when $T,~n \rightarrow \infty$, and the following result follows by Slutsky's Theorem.
\begin{corollary}\label{corollary2}
    Under the assumptions in Theorem~\ref{thm:theorem4} and Condition~\ref{con:condition2}, as $T,~n \rightarrow \infty$, 
    $
    \sqrt{n}(\widehat{\tau}^{*}_{(j_{1} j_{2})} - \tau_{(j_{1} j_{2})})  \overset{d}{\to} N(0,\mathcal{V}_{(j_{1} j_{2})})$.
\end{corollary}
\noindent The proposed outcome measure, $\widehat{F}_{(j_1 j_2)}(\mathcal{D}_2) \mid \widehat{m}(\mathcal{D}_1)$, is an asymptotically unbiased estimator of $F_{(j_{1}j_{2}) }$, satisfying Condition~\ref{con:condition1} provided that $(T - T_0) > r(p + 2)$. 
Condition~\ref{con:condition2} may require a larger $T$ for each subject. Even when it is not satisfied, $\widehat{\tau}^\ast_{(j_{1}j_{2})}$ remains a valid test statistic for the hypotheses in~\eqref{eq:null} under Assumptions~\ref{as:assumption1}--\ref{as:assumption4}, and~\ref{as:assumptionA1}--\ref{as:assumptionA3}, as established in Theorem~\ref{thm:theorem2}.

\subsection{Simultaneous inference}

When inference concerns multiple components in $\boldsymbol{\tau}$ or $\boldsymbol{\tau}^{*}$, dependence among the estimator $\{\widehat{\tau}^{*}_{(j_{1} j_{2})}\}$ induced by within-outcome correlations and the use of shared estimated propensity scores must be properly accounted for. Valid simultaneous inference therefore requires multiple testing procedures that control error rates such as family-wise error rate (FWER) and/or the false discovery proportion (FDP). 
In Sections~\ref{section5: simulation} and~\ref{section6: application}, we employ a step-down procedure with augmentation that controls FDP exceedance rate. This approach is introduced by \cite{genovese2006exceedance} and recently combined with a Gaussian multiplier bootstrap by \cite{qiu_unveiling_2024}. Further details are provided in Supplementary Material Section~\ref{sec:multiinference}.

\section{Simulation studies} 
\label{section5: simulation}

In this section, we present results of simulation studies conducted to assess the validity and effectiveness of the proposed method. We generate the data for $n=100$ subjects. For each subject $i = [1:n]$, a vector of baseline covariates $\mathbf{W}_i \in \mathbb{R}^{q}$, with $q=5$, is generated from a multivariate normal distribution. The intervention indicator $Z_i$ is then generated according to $Pr(Z_i=1 \mid \mathbf{W}_i) = \text{logit}(\mathbf{W}_i^\top \bm{\beta})$, where $\bm{\beta} = (0.7, -0.8, 0.5, 0.3, -0.3)^\top $. 
For each subject $i$, we generate two potential outcomes of time series $\widetilde{\mathbf{X}}_{i}(z)$, $z \in \{0, 1\}$, of length $T=1000$ and $p\in\{21,51\}$ from a VAR(1) model. The transition matrix describing the linear dependencies among the time series is partitioned into $3\times3$ blocks of equal size, each of dimension $p/3$. 
Within each diagonal block, the autoregressive coefficients  are sampled from $\mathrm{Unif}(0.72, 0.85)$, while the off-diagonal coefficients within the blocks are generated more sparsely from $\mathrm{Unif}(0.03, 0.08)$ with random signs. The intervention effect is introduced through the between-block autoregressive coefficients. 
Specifically, when $Z_i=1$, a subset of these coefficients is set to values with magnitude $\delta \in \{0.02, ~0.04, ~0.06, ~0.08, ~0.10\}$, whereas when $Z_i=0$, all between-block coefficients are set to zero (i.e., $\delta = 0$). The parameter $\delta$ therefore controls the strength of the intervention effect. 
In addition, baseline covariates $\mathbf{W}_i$ influence the baseline connectivity structure by modifying the transition matrices.
These structural patterns are held fixed across subjects across Monte Carlo replications. 

Using simulated data, effective connectivity is derived using the proposed sample-splitting procedure, intervention effects are estimated using the IPW estimator in~\eqref{eq:IPW}, and simultaneous inference is conducted. Hypothesis testing of the null hypothesis in \eqref{eq:null} is performed at the $\alpha$=0.05 level. The step-down procedure with augmentation targets control of the false discovery exceedance probability, $\mathrm{FDPex} \coloneqq Pr(\mathrm{FDP}>0.1) \leq \alpha$. We compare the performance of four alternative methods: (i) the proposed sample-splitting approach and the IPW estimator with a correctly specified propensity score model; (ii) the proposed sample-splitting approach with an incorrect propensity score model; (iii) no sample splitting, using the second half of the data for both model selection and estimation, with a correctly specified propensity score model; and (iv) no sample splitting with an incorrectly specified propensity score model. For each method, we evaluate FWER and average bias under the global null ($\delta=0$). Under alternative scenarios ($\delta>0$), we report power, FDP, and $\mathrm{FDPex}$, averaged  over 1000 Monte Carlo replications. 

\begin{figure}[!ht]
    \centering
    \includegraphics[width=0.8\linewidth]{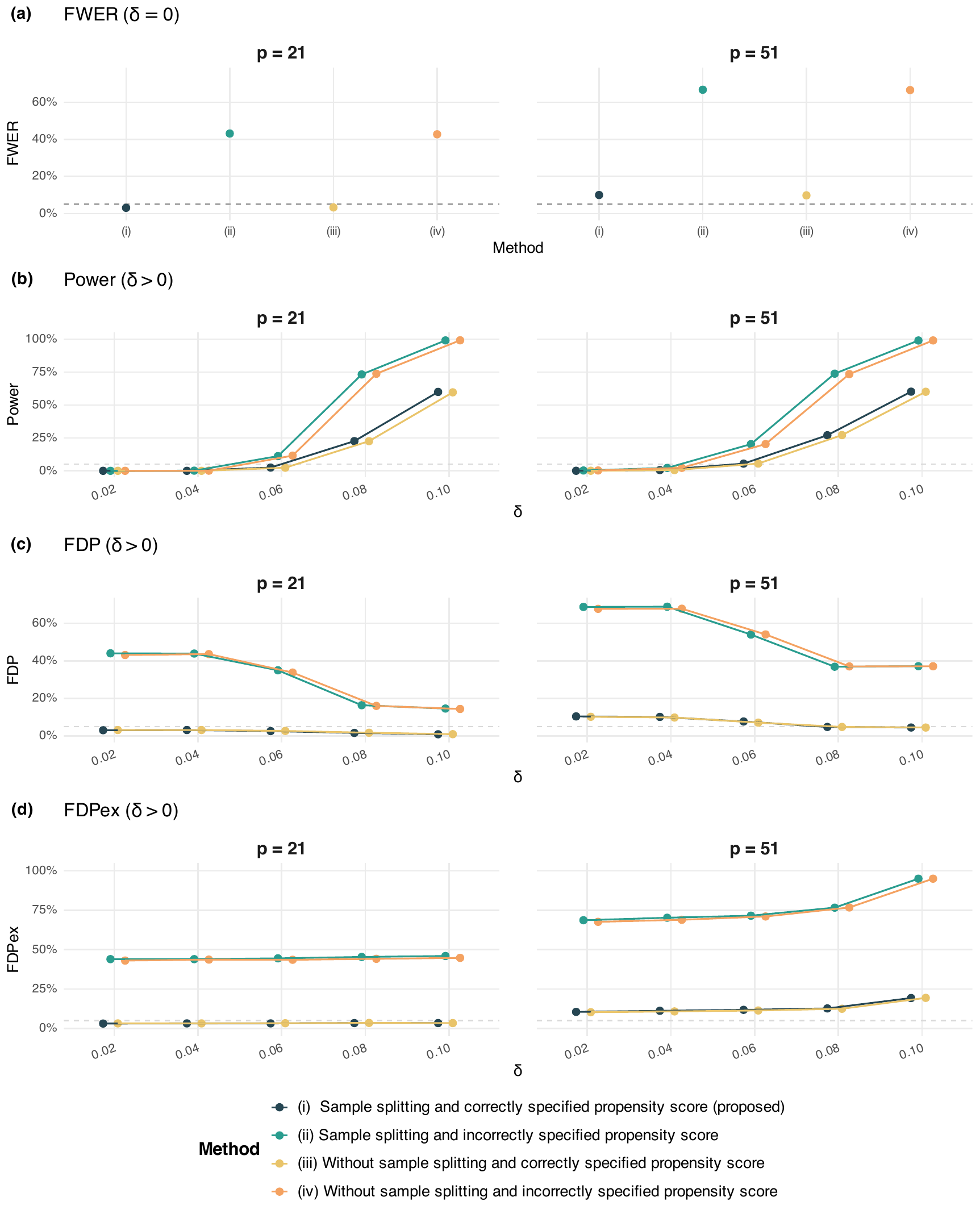}
    \caption{Simulation results when the number of subjects $n=100$. Panel (a) shows the FWER under the global null; panel (b) shows power under alternatives with varying intervention intensities ($\delta$); panels (c) and (d) show the FDP and FDPex, respectively, across different values of $\delta$.}
    \label{fig:simulation_results}
\end{figure}
Figure~\ref{fig:simulation_results} summarizes the simulation results for $n=100$. When $p=21$, the proposed method maintains valid control of the FWER near the nominal $5\%$ level and exhibits increasing power as $\delta$ increases, while controlling the FDP and the $\mathrm{FDPex}$. In contrast, methods based on an incorrectly specified propensity score model exhibit inflated FWER and elevated false discoveries due to biased effect estimation and underestimation of variance. The error rates are further inflated when $p=51$, as the number of direction components increase exponentially compared to $p=21$. Method (iii) performs comparably to the proposed method in FWER, power, FDP and FDPex, which may suggest that sample splitting offers little benefit in this setting. However, similar results likely reflects the limited degree of overfitting in our data-generating mechanism, which is based on relatively simple VAR models. In more complex settings with complicated confounding structures among outcome units, overfitting is expected to be more pronounced, making sample splitting necessary to ensure valid inference. Overall, the proposed method demonstrates valid type-I error control under the global null and improved power as the effect size increases. Additional results for larger sample size ($n=200$) are provided in Supplementary Material Section~\ref{supplementary:simulation}.

\section{Investigation of the causal effect of amyloid elevation on effective connectivity in Alzheimer's Disease} \label{section6: application}

We apply the proposed framework to examine the effect of amyloid burden, evaluated using amyloid PET imaging, on measures of brain effective connectivity. Our study population is restricted to participants in the non-dementia ADNI research groups at entry, including cognitively normal (CN), significant memory concern (SMC), early mild cognitive impairment (EMCI), and mild cognitive impairment (MCI). Accordingly, both the amyloid-positive and amyloid-negative groups exclude dementia participants, whose effective connectivity patterns may differ substantially related to dementia processes. For each subject, resting-state fMRI data are preprocessed and registered to the Montreal Neurological Institute (MNI) 152 T1 2mm standard template space. We compute region-specific mean time series using the AAL2 atlas~\citep{rolls_implementation_2015}, resulting in $p=$120 time series corresponding to predefined ROIs. Additional preprocessing details are provided in Supplementary Material Section~\ref{supplementary:data_application}. 

\begin{figure}[!ht]
    \centering
    \includegraphics[width=0.7\linewidth]{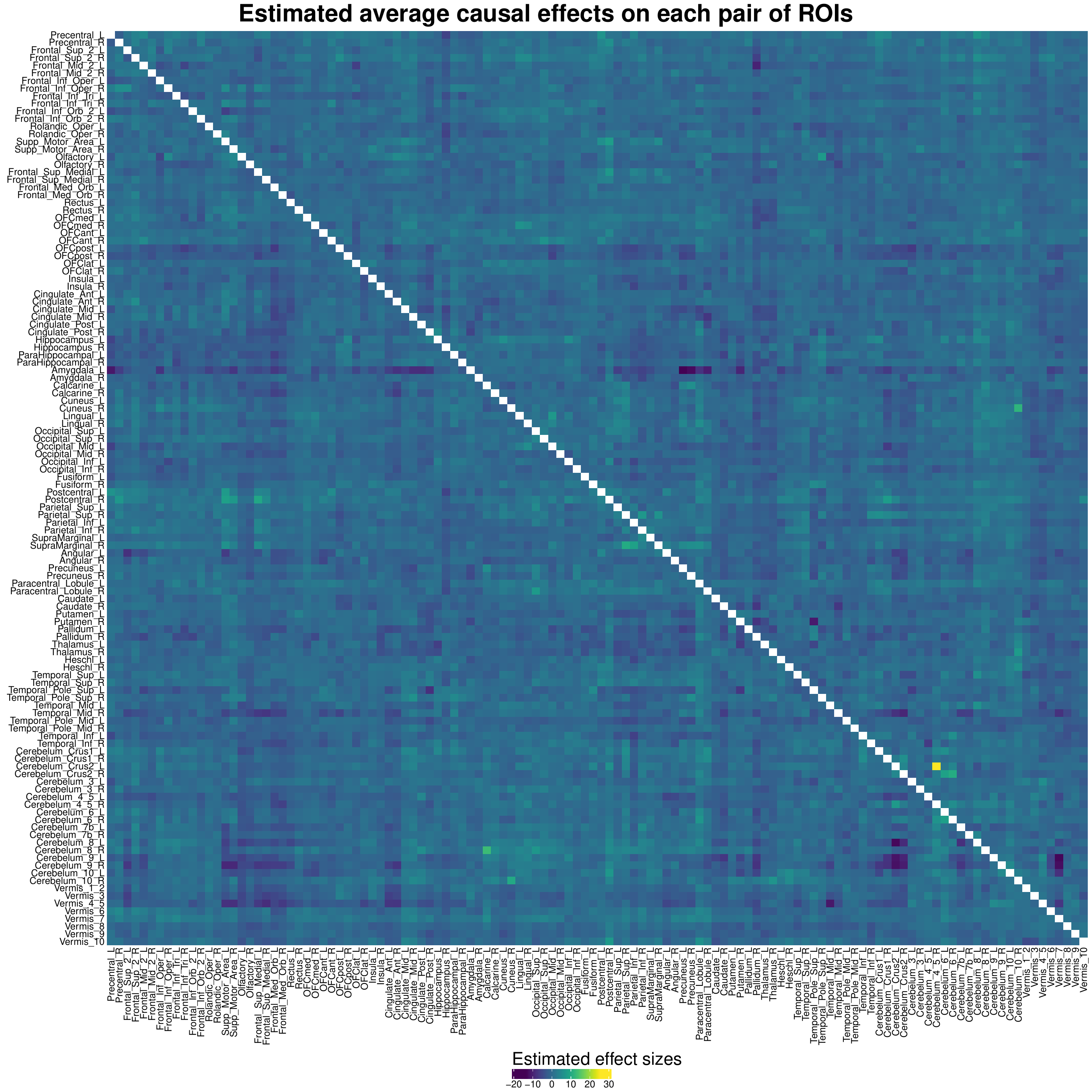}
    \caption{A heat map visualizing the full matrix of estimated causal effects of amyloid positivity on effective connectivity across pairs of ROIs ($p=120$).}
    \label{fig:ate_all}
\end{figure}

Among the $n = 81$ subjects, 37 (45.7\%) are amyloid positive. The amyloid-positive and amyloid-negative groups do not differ significantly with respect to sex (40.5\% vs.\ 36.4\% male, $p=0.876$) or years of education (16.22~$\pm$~2.45 vs. 16.36~$\pm$~2.72~years; $p= 0.659$). However, subjects in the amyloid-positive group are older on average than those in the amyloid-negative group (81.41~$\pm$~7.74 vs 74.66~$\pm$~7.77~years; $p < 0.001$), and the prevalence of APOE~$\varepsilon$4 carrier status is higher among amyloid-positive subjects (54.1\% vs. 27.3\%; $p=0.026$).
Figure~\ref{fig:ate_all} displays the estimated average effects across all outcome unit pairs of effective connectivity. Most estimated effects are close to zero, with a greater proportion taking negative values than positive values. After applying the multiple testing procedure, only a single directed effective connectivity from the left anterior cingulate cortex to the left middle occipital region remains statistically significant, with an estimated effect of $-3.33~( \text{simultaneous } 95\% \text{ CI}: [-6.50, -0.16])$. Overall, these results provide limited evidence for a global reduction in effective connectivity due to elevated amyloid burden. The finding may reflect limitations of resting-state fMRI data, including relatively low resolution and measurement noise. 

\section{Discussion}\label{section7:discussion}

The proposed framework for causal inference with derived outcomes from time-series data is broadly applicable in settings where the outcome of interest represents relationships between units, such as friendship, collaboration, or biological interaction, and is typically inferred from longitudinal observations. The methodology is particularly relevant when the primary research goal is to examine the causal effect of an external intervention on individuals or selected components of these relationship measures, rather than to study causal mechanisms among the relationships themselves, and when the outcomes are derived using well-established models that are widely adopted in the relevant  scientific domains (e.g., VARs). A key principle underlying our approach is that models used to construct relational outcomes should appropriately account for, and condition on, others' observed historical information when the inferential target is the intervention effect on one or more components of the derived outcomes. Specifically, by deriving outcome measures that condition on other time series, the proposed framework mitigates spurious associations between the intervention and individual outcome components, thereby enabling more interpretable and valid downstream causal inference.

However, the proposed framework has several limitations that point to meaningful directions for future research. 
First, the selection of conditioning time series currently relies on data from both intervention groups. Future work could develop procedures that define conditioning sets using only pre-intervention data, thus ensuring a clear separation between outcome construction and causal estimation. 
Second, the proposed variance estimator does not explicitly account for the additional variability introduced by the sample-splitting procedure, although we acknowledge discrepancies between $\widehat{\boldsymbol{\tau}}^{*}$ and $\widehat{\boldsymbol{\tau}}$, as well as between their corresponding variance estimators.   
Finally, our current implementation focuses on linear, weakly stationary time-series models, which are commonly used in connectivity analyses. While the use of widely adopted connectivity measures facilitates integration with existing practice, future work could investigate extensions to more flexible and general time-series-based measures with a principled causal inference framework. 

\section*{Acknowledgments}
Data collection and sharing for this project was funded by the Alzheimer's Disease Neuroimaging Initiative
(ADNI) (National Institutes of Health Grant U01 AG024904) and DOD ADNI (Department of Defense award
number W81XWH-12-2-0012). ADNI is funded by the National Institute on Aging, the National Institute of
Biomedical Imaging and Bioengineering, and through generous contributions from the following: AbbVie,
Alzheimer’s Association; Alzheimer’s Drug Discovery Foundation; Araclon Biotech; BioClinica, Inc.; Biogen;
Bristol-Myers Squibb Company; CereSpir, Inc.; Cogstate; Eisai Inc.; Elan Pharmaceuticals, Inc.; Eli Lilly and
Company; EuroImmun; F. Hoffmann-La Roche Ltd and its affiliated company Genentech, Inc.; Fujirebio; GE
Healthcare; IXICO Ltd.; Janssen Alzheimer Immunotherapy Research \& Development, LLC.; Johnson \&
Johnson Pharmaceutical Research \& Development LLC.; Lumosity; Lundbeck; Merck \& Co., Inc.; Meso
Scale Diagnostics, LLC.; NeuroRx Research; Neurotrack Technologies; Novartis Pharmaceuticals
Corporation; Pfizer Inc.; Piramal Imaging; Servier; Takeda Pharmaceutical Company; and Transition
Therapeutics. The Canadian Institutes of Health Research is providing funds to support ADNI clinical sites
in Canada. Private sector contributions are facilitated by the Foundation for the National Institutes of Health
(www.fnih.org). The grantee organization is the Northern California Institute for Research and Education,
and the study is coordinated by the Alzheimer’s Therapeutic Research Institute at the University of Southern
California. ADNI data are disseminated by the Laboratory for Neuro Imaging at the University of Southern
California.

\section*{Funding}
This research was supported by Grant Number 5P20GM103645 from the National Institute of General Medical Sciences and Grant Number R01AG075511 from the National Institute on Aging.

\bibliographystyle{Chicago}
\bibliography{references}

\newpage
\begin{bibunit}[Chicago]
\appendix

\begin{center}
{\Large
\textbf{Supplementary materials}
}
\end{center}

\setcounter{equation}{0}
\setcounter{figure}{0}
\setcounter{table}{0}
\setcounter{page}{1}
\setcounter{section}{0}
\renewcommand{\theequation}{A\arabic{equation}}
\renewcommand{\thefigure}{A\arabic{figure}}
\renewcommand{\thesection}{A\arabic{section}}
\renewcommand{\thetable}{A\arabic{table}}
\renewcommand{\thefigure}{A\arabic{figure}}

\section{Details of conditional Granger causality models}
The effective connectivity outcome measures are constructed separately for each subject. We therefore suppress the subject index $i$ throughout this section for notational simplicity. 

Let $\mathbf{X}_{\cdot \mathcal{J}}^{(t)}$ denote the length-$\ell$ vector of time series corresponding to outcome units (e.g., brain regions) in the index set $\mathcal{J}$ at time $t$. The reduced VAR system for $X_{\cdot j_2}^{(t)}$, conditional on $\mathbf{X}_{\cdot \mathcal{J}}^{(t)}$ for $t=[(r+1):T]$, can be represented as follows: 
\begin{eqnarray} \label{eq:var1}
\begin{pmatrix}
    X_{\cdot j_2}^{(t)} \\
    \mathbf{X}_{\cdot \mathcal{J}}^{(t)}
\end{pmatrix}
= \mathbf{C}^{(t)} + \sum_{k=1}^r \bm{A}_k \begin{pmatrix}
    X_{\cdot j_2}^{(t-k)} \\
    \mathbf{X}_{\cdot \mathcal{J}}^{(t-k)}
\end{pmatrix} + 
\begin{pmatrix}
    \varepsilon_{j_2}^{(t)} \\
    \boldsymbol{\varepsilon}_{\mathcal{J}}^{(t)}
\end{pmatrix},
\end{eqnarray}
where $ \mathbf{C}^{(t)} $ is a vector of length of $l+1$ that includes an intercept, and $\bm{A}_k$ are the $(l+1) \times (l+1)$ transition matrices for lag $k=[1:r]$. The vector $\boldsymbol{\varepsilon}^{(t)}=(\varepsilon_{j_2}^{(t)},(\boldsymbol{\varepsilon}_{\mathcal{J}}^{(t)})^\top)^\top$ includes the error terms at time $t$, which satisfies $\mathbb{E}(\boldsymbol{\varepsilon}^{(t)}) =\mathbf{0}$. Let $\boldsymbol{\Sigma}$ denote the $(l+1) \times (l+1)$ covariance matrix of $\boldsymbol{\varepsilon}^{(t)} $ for $t=[(r+1): T]$. This matrix can be partitioned as 
\begin{eqnarray*}
\boldsymbol{\Sigma}=\mathbb{E}(\boldsymbol{\varepsilon}^{(t)} (\boldsymbol{\varepsilon}^{(t)})^\top)= \begin{pmatrix}
    var(\boldsymbol{\varepsilon}^{(t)}) & cov(\boldsymbol{\varepsilon}^{(t)},\boldsymbol{\varepsilon}_{\mathcal{J}}^{(t)})\\
    cov(\boldsymbol{\varepsilon}_{\mathcal{J}}^{(t)}, \varepsilon_{j_2}^{(t)}) & var(\boldsymbol{\varepsilon}_{\mathcal{J}}^{(t)}).
\end{pmatrix}
\end{eqnarray*}

The full model for $X_{\cdot j_1}^{(t)}, X_{\cdot j_2}^{(t)}$ and $\mathbf{X}_{\cdot \mathcal{J}}^{(t)}$ can be written as
\begin{eqnarray} \label{eq:var2}
\begin{pmatrix}
    X_{\cdot j_1}^{(t)} \\
    X_{\cdot j_2}^{(t)} \\
   \mathbf{X}_{\cdot \mathcal{J}}^{(t)}
\end{pmatrix}
=  \mathbf{C}'^{(t)} + \sum_{k=1}^r \bm{A}'_k \begin{pmatrix}
    X_{\cdot j_1}^{(t-k)} \\
    X_{\cdot j_2}^{(t-k)} \\
    \mathbf{X}_{\cdot \mathcal{J}}^{(t-k)}
\end{pmatrix} + 
\begin{pmatrix}
    {\varepsilon'}_{j_1}^{(t)} \\
    {\varepsilon'}_{j_2}^{(t)} \\
    {\boldsymbol{\varepsilon}'}_{\mathcal{J}}^{(t)}
\end{pmatrix}.
\end{eqnarray}
Here, $\mathbf{C}'^{(t)}$ is a vector of length of $l+2$ that includes an intercept, and $\bm{A}'_k$ are the $(l+2) \times (l+2)$ transition matrices for lags $k=[1:r]$ in the full VAR system. 
The error terms at time $t$ is $\boldsymbol{\varepsilon'^{(t)}}=(\varepsilon'^{(t)}_{j_1}, \varepsilon'^{(t)}_{j_2},(\boldsymbol{\varepsilon}'^{(t)}_{\mathcal{J}})^\top)^\top$, which satisfies $\mathbb{E}(\boldsymbol{\varepsilon'^{(t)}}) = \mathbf{0}$. Let  $\boldsymbol{\Sigma}^{\prime}$ denote the $(l+2) \times (l+2)$ covariance matrix of $\boldsymbol{\varepsilon'^{(t)}}$ for $t = [(r+1):T]$, which can be partitioned as
\begin{eqnarray*}
\boldsymbol{\Sigma}'=\mathbb{E}(\boldsymbol{\varepsilon'^{(t)}}(\boldsymbol{\varepsilon'^{(t)}})^\top) = \begin{pmatrix}
    var(\varepsilon'^{(t)}_{j_1}) & cov( \varepsilon'^{(t)}_{j_1},\varepsilon'^{(t)}_{j_2}) & cov( \varepsilon'^{(t)}_{j_1},\boldsymbol{\varepsilon}'^{(t)}_{\mathcal{J}})\\
    cov(\varepsilon'^{(t)}_{j_2}, \varepsilon'^{(t)}_{j_1}) & var(\varepsilon'^{(t)}_{j_2}) &
    cov(\varepsilon'^{(t)}_{j_2},\boldsymbol{\varepsilon}'^{(t)}_{\mathcal{J}}) \\
    cov(\boldsymbol{\varepsilon}'^{(t)}_{\mathcal{J}}, \varepsilon'^{(t)}_{j_1}) &
    cov(\boldsymbol{\varepsilon}'^{(t)}_{\mathcal{J}},\varepsilon'^{(t)}_{j_2}) & var(\boldsymbol{\varepsilon}'^{(t)}_{\mathcal{J}}). 
\end{pmatrix}
\end{eqnarray*}

Under the following two assumptions, the parameters of each equation in the VAR system in~\eqref{eq:var1} and~\eqref{eq:var2} can be estimated by ordinary least squares. For resting-state fMRI data, the weak-stationarity assumption is commonly adopted and is often reasonable after appropriate preprocessing. 
\begin{bassumption} \label{as:assumptionB1}
The error terms in the VAR models are Gaussian with time-invariant covariance matrices.
    \begin{enumerate}[label=(\roman*)]
        \item Multivariate normality: The error terms satisfy $\boldsymbol{\varepsilon}^{(t)} \sim MVN_{l+1}(\mathbf{0}, \boldsymbol{\Sigma})$ and $\boldsymbol{\varepsilon}'^{(t)} \sim MVN_{l+2}(\mathbf{0}, \boldsymbol{\Sigma}')$.
        \item Noises serial correlation: The covariance matrices $\boldsymbol{\Sigma}$ and $\boldsymbol{\Sigma}'$ do not depend on $t$.
    \end{enumerate}
\end{bassumption}

\begin{bassumption}[Weak Stationarity] \label{as:assumptionB2}
   The time series $X_{\cdot j_1}^{(t)}, X_{\cdot j_2}^{(t)}$, and each series in $\mathbf{X}_{\cdot \mathcal{J}}^{(t)}$ are weakly stationary. That is, each has constant mean  and finite second moment, and the autocovariance function depend only on lag. 
\end{bassumption}

\section{Lemmas and Proofs}
\label{sec:proof}

\begin{lemma}[Tightness] \label{lemma: lemma1}
    For conditional Granger causality, the test statistic is an F-statistic . Let $\widehat{F}_{(j_1 j_2)}(\mathcal{D}_2) \mid \widehat{m}(\mathcal{D}_1)$ denote  the estimator of this statistic computed from $\mathcal{D}_{2}$ (with time length $T-T_0$), conditional on the model fitted on $\mathcal{D}_{1}$. Under the null hypothesis that $\widetilde{\mathbf{X}}_{\cdot j_1}$ does not Granger-cause $\widetilde{\mathbf{X}}_{\cdot j_2}$ conditioning on $\widetilde{\mathbf{X}}_{\cdot \mathcal{J}}$,
\begin{eqnarray*}
 \widehat{F}_{(j_1 j_2)}(\mathcal{D}_2) \mid \widehat{m}(\mathcal{D}_1) \sim F(r,(T-T_0-r)-(\ell+2)r-1).
 \end{eqnarray*}
Since $r$ and $\ell$ are fixed, as $T - T_0 \rightarrow \infty$, $\widehat{F}_{(j_1 j_2)}(\mathcal{D}_2) \mid \widehat{m}(\mathcal{D}_1) \overset{d}{\rightarrow} \chi^2_r/r$, under the null hypothesis.  
\end{lemma}
\noindent Because convergence in distribution to a non-degenerate limit implies tightness, the sequence of conditional Granger causality estimators is tight. That is, for any $\delta > 0 $, there exists a compact interval such that the probability that $\widehat{F}_{(j_1 j_2)}(\mathcal{D}_2) \mid \widehat{m}(\mathcal{D}_1)$ lies outside this interval is smaller than $\delta$.

\begin{proof}[Proof of Theorem~\ref{thm:theorem1}]

\begin{eqnarray*}   
   \tau^{*} & = & \mathbb{E}_{p} \left[ \mathbb{E}_{p}(\widehat{\mathbf{Y}}_i(1) \mid  \mathbf{W}_i) - \mathbb{E}_{p}(\widehat{\mathbf{Y}}_i (0) \mid \mathbf{W}_i) \right] \quad \text{(Assumption \ref{as:assumption4})}\\
    & = & \mathbb{E}_{p}\left[ \mathbb{E}_{p}(\mathbb{E}_{i}[\widehat{\mathbf{Y}}_i(1) \mid  \widetilde{\mathbf{X}}_i(1), \mathbf{W}_i]) - \mathbb{E}_{p}(\mathbb{E}_{i}[\widehat{\mathbf{Y}}_i(0) \mid \widetilde{\mathbf{X}}_i(0), \mathbf{W}_i]) \right]\\
    & = & \mathbb{E}_{p}\left[ \mathbb{E}_{p}(\mathbf{Y}_{i}(1)+\boldsymbol{\Delta}_{i}(1)  \mid  \mathbf{W}_i) - \mathbb{E}_{p}(\mathbf{Y}_{i}(0)+\boldsymbol{\Delta}_{i}(0) \mid \mathbf{W}_i) \right] \\
    & = & \mathbb{E}_{p}[ \mathbb{E}_{p}(\mathbf{Y}_{i}(1) - \mathbf{Y}_{i}(0)  \mid  \mathbf{W}_i)] + \mathbb{E}_{p}[\mathbb{E}_{p}(\boldsymbol{\Delta}_{i}(1)-\boldsymbol{\Delta}_{i}(0) \mid \mathbf{W}_i)]\\ 
   & = & \mathbb{E}_{p}[\mathbf{Y}_{i}(1) - \mathbf{Y}_{i}(0)] + o_{p}(1) \quad \text{(Condition \ref{con:condition1})}.
\end{eqnarray*}
\end{proof}

\begin{proof}[Proof of Theorem~\ref{thm:theorem3}]
Under Assumption~\ref{as:assumptionB1}, $\varepsilon_{j_2}^{(t)} \sim N(0, var(\varepsilon_{j_2}^{(t)}))$, and $\varepsilon'^{(t)}_{j_2} \sim N(0, var(\varepsilon'^{(t)}_{j_2}))$, with finite variances. Thus, $\mathbb{E}[(\varepsilon_{j_2}^{(t)})^2]=var(\varepsilon_{j_2}^{(t)}), \mathbb{E}[(\varepsilon'^{(t)}_{j_2})^2]=var(\varepsilon'^{(t)}_{j_2})$. By Stein's Lemma, $var[(\varepsilon_{j_2}^{(t)})^2] =2var^2(\varepsilon_{j_2}^{(t)}) < \infty$, and similarly for $var[\varepsilon'^{(t)}_{j_2})^2] < \infty$. 

Assume that the VAR models are correctly specified. Under Assumptions~\ref{as:assumptionB1} and~\ref{as:assumptionB2}, and law of large numbers implies that, as $T - T_0 \rightarrow \infty$ and $T - T_0-L_{T - T_0} \rightarrow \infty$,
\begin{eqnarray*}
\frac{\sum_{t=T_0+r+1}^{T} (\widehat{\varepsilon}_{j_2}^{(t)})^2}{(T-T_0-r)-(\ell+2)r-1} &\overset{p}{\rightarrow}& var(\varepsilon_{j_2}^{(t)}), \\   \frac{\sum_{t=T_0+r+1+L_{T - T_0}}^{T} (\widehat{\varepsilon}_{j_2}^{(t)})^2}{(T-T_0-r-L_{T - T_0})-(\ell+2)r-1} &\overset{p}{\rightarrow}
& var(\varepsilon_{j_2}^{(t)}),
\end{eqnarray*}
and likewise,
\begin{eqnarray*}
\frac{\sum_{t=T_0+r+1}^{T} (\widehat{\varepsilon}'^{(t)}_{j_2})^2}{(T-T_0-r)-(\ell+2)r-1} &\overset{p}{\rightarrow}&
var(\varepsilon'^{(t)}_{j_2}), \\ \frac{\sum_{t=T_0+r+1+L_{T - T_0}}^{T} (\widehat{\varepsilon}'^{(t)}_{j_2})^2}{(T-T_0-r-L_{T - T_0})-(\ell+2)r-1} &\overset{p}{\rightarrow}& var(\varepsilon'^{(t)}_{j_2}).
\end{eqnarray*}

Consider the ratio of the two conditional $F$-statistics, each represented as a function of the quantities above. After algebraic rearrangement, it can be written as a product of ratios involving sample residual variances and degree-of-freedom adjustments. Each component converges in probability on one by the law of large numbers and the continuous mapping theorem. Consequently, 
\begin{eqnarray*}
  \frac{\widehat{F}_{(j_1 j_2)}(\mathcal{D}_2) \mid \widehat{m}(\mathcal{D}_1)}{\widehat{F}_{(j_1 j_2)}(\mathcal{D}'_2) \mid \widehat{m}(\mathcal{D}_1)} 
& \overset{p}{\longrightarrow} & 1
\end{eqnarray*}
as $ T - T_0 \rightarrow \infty$ and $T - T_0-L_{T - T_0} \rightarrow \infty$. 
\end{proof}

\begin{proof}[Proof of Theorem~\ref{thm:theorem4}]

Let $Z_i$ follow a Bernoulli distribution with  $Pr(Z_i=1 \mid \mathbf{W}_i) \coloneqq \sigma(\eta_i)$, where $\eta_i = \mathbf{W}_i^\top \bm{\beta}$ and $\sigma(\cdot)$ is a smooth link function mapping $\eta_i$ to $(0,1)$. Define $\sigma_i = \sigma(\eta_i)$, $\sigma_i' = \sigma'(\eta_i)$, and $\omega_i = \dfrac{\sigma_i'}{\sigma_i(1-\sigma_i)}$. The estimator $\widehat{\bm{\beta}}$ denotes the maximum likelihood estimator (MLE) of $\bm{\beta}$, obtained by maximizing 
\begin{eqnarray*}
\sum_{i=1}^{n} \left\{ Z_i \log\sigma_i + (1-Z_i) \log(1-\sigma_i)\right\}.
\end{eqnarray*}
The corresponding score function is
\begin{eqnarray*}
\mathbf{S} = \sum_{i=1}^{n} (Z_i - \sigma_i)\omega_i\mathbf{W}_i,
\end{eqnarray*}
and the Fisher information matrix is
\begin{eqnarray*}
\mathcal I_{\bm{\beta}} = 
\mathbb{E} \left[
\frac{(\sigma_i')^2}{\sigma_i(1-\sigma_i)}\,\mathbf W_i \mathbf W_i^\top
\right].
\end{eqnarray*}
\noindent Under standard regularity conditions, 
\begin{equation}  \label{eq:expan_beta}
\widehat{\bm{\beta}} - \bm{\beta} 
= \mathcal I_{\bm{\beta}}^{-1}
\frac{1}{n}\sum_{i=1}^n (Z_i - \sigma_i)\,\omega_i\,\mathbf W_i
+ o_{p}(n^{-1/2}).
\end{equation}
\noindent The corresponding IPW estimator is 
\begin{align*}
    \widehat{\tau}^{*}_{(j_{1} j_{2})} = \frac{1}{n}\sum_{i=1}^n \left\{ \frac{Z_i \widehat{Y}_{i (j_{1} j_{2})}}{\sigma(\mathbf{W}_i^\top \widehat{\bm{\beta}})} - \frac{(1-Z_i)\widehat{Y}_{i (j_{1} j_{2})}}{1-\sigma(\mathbf{W}_i^\top \widehat{\bm{\beta}})} \right\},
\end{align*}
which depends on $\widehat{\bm{\beta}}$. Define 
\begin{eqnarray*}
T(\bm{\beta}) = \frac{1}{n}\sum_{i=1}^n 
\left\{
\frac{Z_i \widehat{Y}_i}{\sigma_i}
-
\frac{(1-Z_i)\widehat{Y}_i}{1-\sigma_i}
\right\}.
\end{eqnarray*}
Its gradient is
\begin{equation} \label{eq:gradient_Tbeta}
\nabla_{\bm{\beta}} T(\bm{\beta})
= \frac{1}{n}\sum_{i=1}^n 
\left[
-\,Z_i\widehat{Y}_i\,\frac{\sigma_i'}{\sigma_i^2}
- (1-Z_i)\widehat{Y}_i\,\frac{\sigma_i'}{(1-\sigma_i)^2}
\right]\mathbf W_i.
\end{equation}
A first-order Taylor expansion around $\bm{\beta}$ gives 
\begin{align*}
\widehat{\tau}^{*}_{(j_{1} j_{2})} -  \frac{1}{n}\sum_{i=1}^n \left\{ \frac{Z_i \widehat{Y}_{i (j_{1}j_{2})}}{\sigma(\mathbf{W}_i^\top \bm{\beta})} - \frac{(1-Z_i)\widehat{Y}_{i (j_{1} j_{2})}}{1-\sigma(\mathbf{W}_i^\top \bm{\beta})} \right\}= \nabla_{\bm{\beta}} T(\bm{\beta})^\top (\widehat{\bm{\beta}} - \bm{\beta}) + o_{p}(n^{-1/2}).
\end{align*}
Substituting the expansions the above yields
\begin{eqnarray*}   
\widehat{\tau}^{*}_{(j_{1} j_{2})} - \tau^{*}_{(j_{1} j_{2})}     &=& 
\frac{1}{n}\sum_{i=1}^{n} h_{i(j_{1} j_{2})} + o_{p}(n^{-1/2}),
\end{eqnarray*}
where
\begin{align*} 
h_{i(j_{1} j_{2})} & = \frac{Z_i \widehat{Y}_{i (j_{1} j_{2})}}{\sigma(\mathbf{W}_i^\top \bm{\beta})} - \frac{(1-Z_i)\widehat{Y}_{i (j_{1} j_{2})}}{1-\sigma(\mathbf{W}_i^\top \bm{\beta})} - \tau^{*}_{(j_{1} j_{2})}\\
& - \left[ \frac{1}{n}\sum_{i=1}^n 
\left\{Z_i\widehat{Y}_{i (j_{1} j_{2})}\frac{\sigma_i'}{\sigma_i^2}
+ (1-Z_i)\widehat{Y}_{i (j_{1} j_{2})} \frac{\sigma_i'}{(1-\sigma_i)^2}
\right\}\mathbf W_i^\top \right] \mathcal I_{\bm{\beta}}^{-1}
(Z_i - \sigma_i) \omega_i \mathbf W_i.
\end{align*}
This linear expansion implies that $\widehat{\tau}^{*}_{(j_{1} j_{2})}$ is asymptotically normal with the influence function $h_{i(j_{1} j_{2})}, i=[1:n]$. If the propensity score model is correctly specified, when under Assumptions \ref{as:assumption1}--\ref{as:assumption4}, we have $\sqrt{n}(\widehat{\tau}^{*}_{(j_{1} j_{2})} - \tau^{*}_{(j_{1} j_{2})})  \overset{d}{\to} N(0,\mathcal{V}_{(j_{1} j_{2})})$, as $n \rightarrow \infty$, where $\mathcal{V}_{(j_{1} j_{2})} = \mathbb{E}(n^{-1}\sum_{i=1}^{n}h_{i(j_{1} j_{2})}^2)$.

\end{proof}

The proof of Theorem~\ref{thm:theorem4} follows similar arguments to those in~\cite{qiu_unveiling_2024}. In this work, the proposed IPW estimator based on the derived multivariate outcomes $\widehat{\mathbf{Y}}$ targets the estimand defined in terms of the potential outcomes for the unobserved outcome $\mathbf{Y}$, denoted by $\boldsymbol{\tau}$. Our inference further targets the estimand defined with respect to the potential outcomes associated with the derived outcomes $\widehat{\mathbf{Y}}_{i}$, denoted by $\boldsymbol{\tau}^{\ast}$. Moreover, the link function $\sigma(\cdot)$ can be any smooth mapping from $\eta_i$ to $(0,1)$, and is not restricted to the logistic form.

\begin{Aremark}
When $\sigma(u) = 1/(1+e^{-u})$, we have $\sigma'(u)=\sigma(u)\{1-\sigma(u)\}$, so that $\omega_i = 1$. The expression for $h_{i(j_{1} j_{2})}$ above then reduces exactly to the standard logistic regression case: 
\begin{eqnarray*}
&& \frac{Z_i \widehat{Y}_{i (j_{1} j_{2})}}{\sigma(\mathbf{W}_i^\top \bm{\beta})} - \frac{(1-Z_i)\widehat{Y}_{i (j_{1} j_{2})}}{1-\sigma(\mathbf{W}_i^\top \bm{\beta})} - \tau^{*}_{(j_{1} j_{2})}\\
& -&  \left[ \frac{1}{n} \sum_{i=1}^{n}\left\{\frac{Z_i \widehat{Y}_{i (j_{1} j_{2})}(1-\sigma(\mathbf{W}_i^\top \bm{\beta}))}{\sigma(\mathbf{W}_i^\top \bm{\beta})} + \frac{(1-Z_i) \widehat{Y}_{i (j_{1} j_{2})}\sigma(\mathbf{W}_i^\top \bm{\beta})}{1-\sigma(\mathbf{W}_i^\top \bm{\beta})}\right\}\mathbf{W}_i^\top \right] \mathcal{I}_{\bm{\beta}}^{-1}(Z_i - \sigma(\mathbf{W}_i^\top \bm{\beta}))\mathbf{W}_i.
\end{eqnarray*}
\end{Aremark}

\begin{proof}[Proof of Corollary~\ref{corollary1}]
By the consistency of the MLE $\widehat{\bm{\beta}} \overset{p}{\longrightarrow} \bm{\beta}$ and by continuity of the smooth function $\sigma(\cdot)$, it follows that
$ \sigma(\mathbf{W}_i^\top \widehat{\bm{\beta}})\overset{p}{\longrightarrow} \sigma(\mathbf{W}_i^\top \bm{\beta})$ and $\widehat{\mathcal{I}} \overset{p}{\longrightarrow}  \mathcal{I}_{\bm{\beta}}$. For $i \in [1:n]$, we therefore have $\widehat{h}_{i(j_1j_2)} \overset{p}{\to}  h_{i(j_1 j_2)}$ and $\widehat{h}_{i(j_1 j_2)}^2 \overset{p}{\longrightarrow}  h_{i(j_1 j_2)}^2$. Thus, the plug-in variance estimator $\widehat{\mathcal{V}}_{(j_1 j_2)} = \frac{1}{n} \sum_{i=1}^n \widehat{h}_{i(j_1 j_2)}^2$ is consistent for $\mathcal{V}_{(j_1 j_2)} = \mathbb{E}\{ h_{i(j_1  j_2)}^2 \}$.
\end{proof}

\section{Details about simultaneous inference}
\label{sec:multiinference}

After obtaining Gaussian approximations for each component of $\boldsymbol{\tau}^{\ast}$, we proceed to conduct simultaneous inference on $\boldsymbol{\tau}^{\ast}$. Due to the high dimensionality and complex dependency structures among its components, we apply the Gaussian multiplier bootstrap procedure~\citep{chernozhukov2013}. 

Let $\mathcal{M}$ denote the maximum absolute standardized coordinate $\widehat{\boldsymbol{\tau}}^{*}_{(j_{1} j_{2})}$ over indices $(j_{1}, j_{2})$ with $\widehat{\mathcal{V}}_{(j_{1} j_{2})} > c$ for some pre-specified $c>0$:
\begin{equation}
    \mathcal{M} = \max_{(j_1, j_2):\, \widehat{\mathcal{V}}_{(j_1 j_2)} > c} \left| \sqrt{n}(\widehat{\tau}^{*}_{(j_{1} j_{2})} - \tau^{*}_{(j_{1} j_{2})}) \, \widehat{\mathcal{V}}_{(j_{1} j_{2})}^{-1/2} \right|.
\end{equation}
To approximate the distribution of $\mathcal{M}$, generate i.i.d. standard normal multiplies $\xi_{i,b},~i \in [1:n]$ and bootstrap replication $b \in [1:B]$, and calculate
\begin{equation}
    \mathcal{M}_{b} = \max_{(j_1, j_2):\, \widehat{\mathcal{V}}_{(j_1 j_2)} > c} \left| \frac{1}{\sqrt{n}} \sum_{i=1}^{n} \xi_{i,b} \widehat{h}_{i(j_1 j_2)} \, \widehat{\mathcal{V}}_{(j_1 j_2)}^{-1/2} \right|,
\end{equation}
Let $q_{1-\alpha}$ denote the $(1 - \alpha)$- quantile of $\{ \mathcal{M}_{b}, b \in [1:B]\}$. Simultaneous $(1 - \alpha)$ confidence intervals for $\boldsymbol{\tau}^{*}$ are then given by
\begin{equation}
    \left[ \widehat{\tau}^{*}_{(j_1 j_2)} - \frac{q_{1-\alpha} \widehat{\mathcal{V}}_{(j_1 j_2)}^{1/2}}{\sqrt{n}},~ \widehat{\tau}^{*}_{(j_1 j_2)} + \frac{q_{1-\alpha} \widehat{\mathcal{V}}_{(j_1 j_2)}^{1/2}}{\sqrt{n}} \right], 
\end{equation}
for all $(j_1, j_2)$ such that $\widehat{\mathcal{V}}_{(j_1 j_2)} > c$.

Let $\mathcal{S}_0$ denote the initial set of null hypotheses corresponding to components $(j_1, j_2)$ with $\widehat{\mathcal{V}}_{(j_1 j_2)} > c$ and set $\mathcal{R}_0=\emptyset$. For each iteration $k \geq 1$, define
\begin{eqnarray*}
\mathcal{M}^{(k)} = \max_{(j_1, j_2) \in \mathcal{S}_{k-1}} \left| \sqrt{n} \, \widehat{\tau}^{*}_{(j_1 j_2)} \, / \, \widehat{\mathcal{V}}_{(j_1 j_2)}^{1/2} \right|,
\end{eqnarray*}
and let $(j^*_1, j^*_2)^{(k)}$ denote the index attaining this maximum. Let $q^{(k)}_{1-\alpha}$ be the $(1 - \alpha)$-quantile obtained from the Gaussian multiplier bootstrap procedure. If $\mathcal{M}^{(k)} > q^{(k)}_{1-\alpha}$, update
\begin{eqnarray*}
\mathcal{R}_k = \mathcal{R}_{k-1} \cup \{ (j^*_1, j^*_2)^{(k)} \}, \quad \mathcal{S}_k = \mathcal{S}_{k-1} \setminus \{ (j^*_1, j^*_2)^{(k)} \};
\end{eqnarray*}
otherwise, terminate the step-down process. Let  $\mathcal{R}_{k_0}$ denote the resulting rejection set (discovery set) for some $k_0 \geq 1$. 
We then augment $\mathcal{R}_{k_0}$ by adding the next $\left\lfloor d \, |\mathcal{R}_{k_0}| / (1 - d) \right\rfloor$ most significant remaining hypotheses in $\mathcal{S}_{k_0}$, ranked by the absolute standardized test statistics $\left| \sqrt{n} \widehat{\tau}^{*}_{(j_1 j_2)} / \widehat{\mathcal{V}}_{(j_1 j_2)}^{1/2} \right|$.  Let $\mathcal{R}$ denote the final discovery set after combining the step-down and augmentation steps. As shown in \cite{qiu_unveiling_2024}, this procedure can control the FDP exceedance rate at $\alpha$ and has a power converging to 1, provided that all non-zero $\tau^{\ast}_{(j_1 j_2)}$ have non-degenerated asymptotic variances. 

\section{Details and additional simulation results}\label{supplementary:simulation}
\subsection{Details of the data-generating procedure}

We consider multivariate time series of dimension $p\in\{21,51\}$. For each subject $i=[1:n]$, we generate two potential outcomes of length $T=1000$, denoted by $\widetilde{\mathbf{X}}_i(z)$ for $z\in \{0,1\}$, from a VAR(1) model. The transition matrix follows a block structure, which is partitioned evenly into $3\times3$ blocks, each of size $p/3$:
$$ \bm{A} = 
\begin{pmatrix}
    \bm{A}_{11} &  \bm{A}_{12} & \bm{A}_{13}\\
    \bm{A}_{21} &  \bm{A}_{22} &  \bm{A}_{23} \\
    \bm{A}_{31} &  \bm{A}_{32} & \bm{A}_{33}
\end{pmatrix}.
$$
\noindent The diagonal blocks $(\bm{A}_{11}, \bm{A}_{22}, \bm{A}_{33})$ represent within-hub autoregressive relationships for three interconnected groups of brain regions (hubs). Within each diagonal block, diagonal entries are drawn from $\mathrm{Unif}(0.72, 0.85)$, whereas off-diagonal entries are generated from $\mathrm{Unif}(0.03, 0.08)$ with randomly assigned signs. 
The off-diagonal blocks represent connectivity between different hubs. When $Z_i=1$, a subset of these entries is replaced by values of magnitude $\delta \in \{0.02, 0.04, 0.06, 0.08, 0.10\}$, representing the strength of the intervention effect. When $Z_i=0$, all entries in the off-diagonal blocks are set to zero ($\delta = 0$).
Baseline connectivity is further allowed to depend on subject-level covariates $\mathbf{W}_i$, which modifying the transition matrices. After generating, each transition matrix is rescaled to satisfy weak stationarity. For each simulation replication, the patterns of non-zero entries in the transition matrices is fixed across subjects. 

\subsection{Additional simulation study results}

We examine estimation bias under the null ($\delta=0$), where the true effect is supposed to be zero for all components.

\begin{table}[!ht]
\centering
\begin{tabular}{l|rr|rr}
\toprule
Bias & \multicolumn{2}{c|}{$n=100$} &  \multicolumn{2}{c}{$n=200$} \\ 
& $p=21$ & $p=51$ & $p=21$ & $p=51$ \\
\midrule
(i) Sample splitting + correct PS  & 0.012 & 0.093 & -0.001 & 0.078 \\
(ii) Sample splitting + incorrect PS & 0.326 & 0.664 & 0.323  & 0.659 \\
(iii) No sample splitting + correct PS & 0.012 & 0.093 & -0.001 & 0.078 \\
(iv) No sample splitting + incorrect PS  & 0.326 & 0.664 & 0.323  & 0.659 \\
\bottomrule
\end{tabular}

\vspace{0.8em}
\begin{tabular}{l|rr|rr}
\toprule
RMSE & \multicolumn{2}{c|}{$n=100$} &  \multicolumn{2}{c}{$n=200$} \\ 
 & $p=21$ & $p=51$ & $p=21$ & $p=51$ \\
\midrule
(i) Sample splitting + correct PS   & 0.693 & 0.775 & 0.588 & 0.666 \\
(ii) Sample splitting + incorrect PS  & 0.970 & 1.190 & 0.946 & 1.174 \\
(iii) No sample splitting + correct PS & 0.693 & 0.775 & 0.588 & 0.666 \\
(iv) No sample splitting + incorrect PS  & 0.970 & 1.190 & 0.946 & 1.174 \\
\bottomrule
\end{tabular}
\caption{ \label{tab:bias_rmse} Bias  (top panel) and RMSE (bottom panel) of the estimated effective connectivity components under the global null hypothesis ($\delta=0$). Bias and RMSE are averaged across ROIs, subjects, and Monte Carlo replications.}
\end{table}

Table~\ref{tab:bias_rmse} show the mean bias and the root mean squared error (RMSE) averaged across ROIs, subjects, and Monte Carlo replications. Across all settings, the proposed estimator exhibits mean bias close to zero and relatively small RMSE, both of which decrease as the sample size increases from $n=100$ to $n=200$. These findings indicate that, the proposed procedure provides an asymptotically unbiased effect estimator under the null. Methods~(ii) and~(iv), without correctly specified propensity scores, show substantially inflated biases across all $(p,n)$ combinations.  Method~(iii), which uses the second half of the time series both for selecting conditions sets and deriving the causal estimator, performs similarly to the proposed method in all three bias metrics. This may be due to limited degree of overfitting in our data-generating mechanism, which is based on relatively simple VAR models. As expected, both absolute bias and RMSE increases with dimension ($p$), while increasing the sample size from $n=100$ to $n=200$ reduces variability. Overall, under the global null, the proposed method remains approximately unbiased.

\begin{figure}[!ht]
    \centering
    \includegraphics[width=0.8\linewidth]{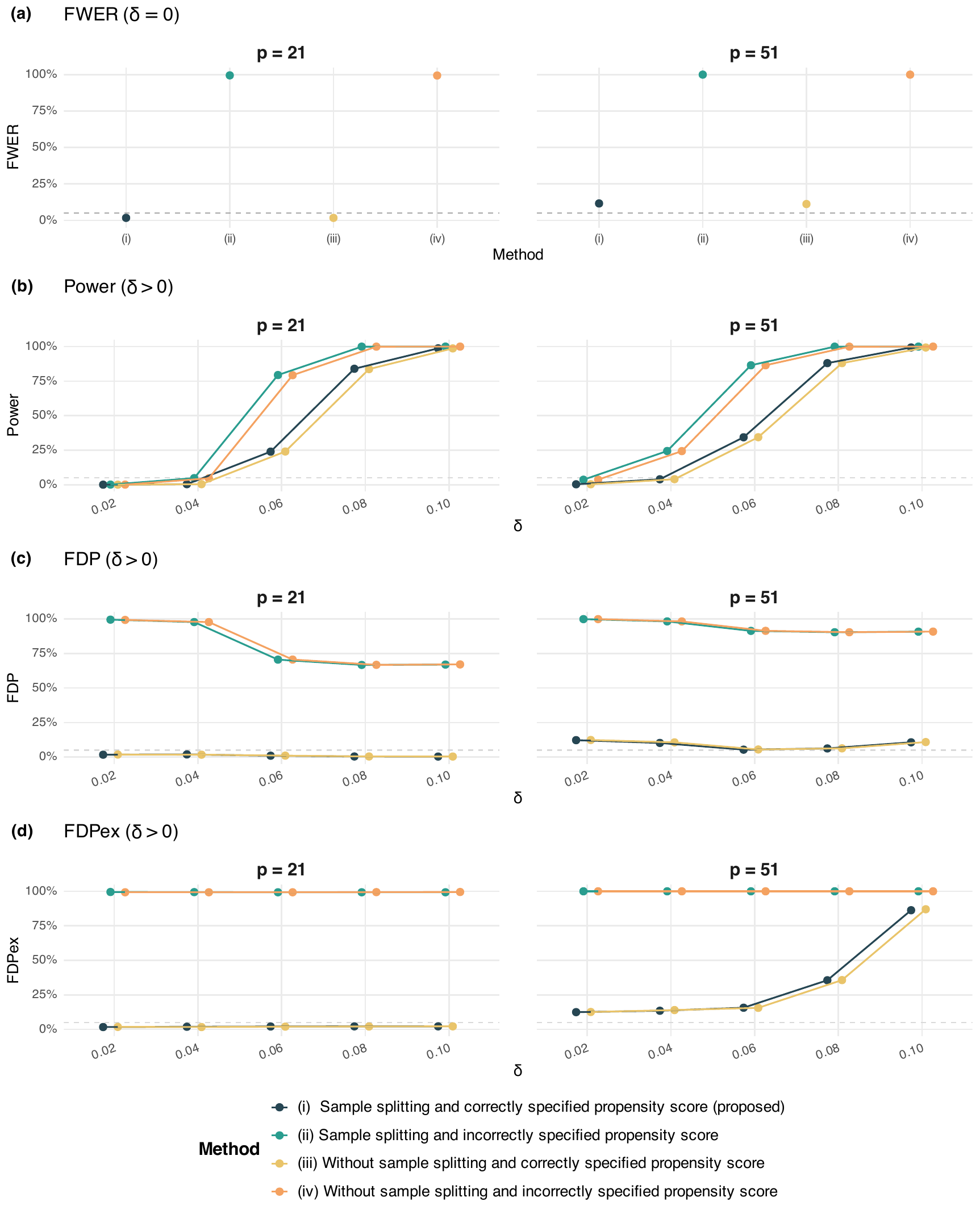}
    \caption{Simulation results when the number of subjects $n=200$. Panel (a) shows the FWER under global null; panel (b) shows power under alternatives with varying intervention intensities ($\delta$); panels (c) and (d) show the FDP and FDPex, respectively, across different values of $\delta$.}
    \label{fig:simulation_results_appendix}
\end{figure}

Figure~\ref{fig:simulation_results_appendix} shows the simulation results of FWER, power, FDP, FDPex across different settings when $n=200$. The patterns observed for $n=100$ (shown in Figure~\ref{fig:simulation_results} in the main manuscript) remain largely unchanged, while the power increases when $n=200$.

\section{Details and additional data application results}\label{supplementary:data_application}
\subsection{Inclusion criteria}
Imaging and clinical data are obtained from phases 3 and 4 of the ADNI. We include participants with resting-state fMRI and amyloid PET scans acquired at the same study visit to ensure temporal alignment between the intervention and functional imaging outcomes. For each resting-state fMRI scan, a concurrent T1-weighted structural MRI within one year is required to facilitate accurate spatial registration. Participants are restricted to those classified at entry as cognitively normal (CN), significant memory concern (SMC), early mild cognitive impairment (EMCI), or mild cognitive impairment (MCI). Accordingly, both the amyloid-positive and amyloid-negative groups exclude dementia participants, whose effective connectivity patterns may differ substantially and may be related to dementia processes.

\subsection{Details of fMRI pre-processing}

Resting-state fMRI data are preprocessed using procedures adapted from \cite{yan2010dparsf}. Raw imaging in the Digital Imaging and Communications in Medicine (DICOM) format are converted to Neuroimaging Informatics Technology Initiative (NIfTI) format. For each scan, the first 10 volumes are discarded to allow for signal equilibration and participants' adaptation to scanner noise. Slice-timing correction is then applied using FSL (slicetimer), followed by rigid-body motion correction with FSL (mcflirt), which also generates six motion parameters for later nuisance regression. To improve computational stability and harmonize time series length across participants, motion-corrected data are truncated to the first 197 time points when longer acquisitions were available. A reference volume is extracted from the motion-corrected data and corrected for intensity inhomogeneity using N4 bias field correction \citep{tustison_n4itk_2010}, followed by skull stripping. For each participant, a concurrent T1-weighted structural MRI acquired closest in time to the resting-state fMRI scan is selected and similarly processed with N4 bias correction and skull stripping. Spatial normalization is performed using a two-stage registration strategy implemented in Advanced Normalization Tools (ANTs). First, the reference image is nonlinearly registered to the participant's T1-weighted image using symmetric diffeomorphic normalization (SyN) \citep{avants_symmetric_2008}. Second, the T1-weighted image is registered to the MNI152 T1 2 mm standard template, and the resulting composite transformations are applied to the full resting-state fMRI. The normalized resting-state fMRI data are spatially smoothed with a Gaussian kernel corresponding to a full width at half maximum (FWHM) of approximately 4.7 mm.

Regional mean time series are then extracted using the Automated Anatomical Labeling 2 (AAL2) atlas, resulting in 120 time series of ROIs per scan\citep{rolls_implementation_2015}. Nuisance regression is applied to remove motion-related artifacts. Linear trends are subsequently removed, and the data are temporally band-pass filtered (0.01-0.10 Hz) to isolate low-frequency fluctuations associated with spontaneous neural activity. Finally, weak stationarity of each ROI series is assessed using the augmented Dickey-Fuller (ADF) test, and only runs in which all 120 ROIs satisfied the stationarity condition are retained for downstream analysis. After preprocessing and quality control procedures, the final sample cohort consists of $n=81$ participants. All preprocessing and analysis code is publicly available at \url{github.com/haiyuesong/Causal-Inference-for-Unobservable-Multivariate-Outcomes}.

\putbib[references]
\end{bibunit}

\end{document}